\documentclass[11pt]{article}

\usepackage[margin=1in]{geometry}
\usepackage{amsmath,amssymb,amsthm}
\usepackage{bm}
\usepackage{graphicx}
\usepackage{booktabs}
\usepackage{longtable}
\usepackage{array}
\usepackage{enumitem}
\usepackage[authoryear,round]{natbib}
\usepackage{xcolor}
\usepackage{xr-hyper}
\usepackage[hidelinks]{hyperref}
\makeatletter\AtBeginDocument{\def\hyper@linkfile#1#2#3{#1\Hy@xspace@end}}\makeatother

\newtheorem{theorem}{Theorem}

\newtheorem{corollary}{Corollary}
\theoremstyle{definition}
\newtheorem{assumption}{Condition}

\newcommand{\E}{\mathbb{E}}

\newcommand{\Pn}{\mathbb{P}_n^{w}}
\newcommand{\PN}{\mathbb{P}_N}
\newcommand{\Pz}{P_0}
\newcommand{\Pzcl}{\mathcal{P}_0^{\mathrm{cl}}}
\newcommand{\D}{\mathrm{D}}
\newcommand{\Hg}{H_g}

\newcommand{\Nhat}{\widehat N}
\newcommand{\psihat}{\widehat\psi_w}

\newcommand{\od}{\mathrm{o}}
\newcommand{\swor}{\sigma_{m,\mathrm{wor}}}


\newcommand{\AppCovMax}{97.4}
\newcommand{\AppCovMcseMax}{1.0}
\newcommand{\AppCovMcseMin}{0.7}
\newcommand{\AppCovMin}{94.4}
\newcommand{\AppFpcBound}{\ensuremath{1.33}}
\newcommand{\AppFpcBoundSe}{\ensuremath{1.15}}

\newcommand{\AppHtsFacilities}{775}

\newcommand{\AppHtsOverlapSens}{7.9}
\newcommand{\AppHtsPsu}{195}

\newcommand{\AppHtsRecords}{4{,}833}

\newcommand{\AppHtsTmleRef}{\ensuremath{-4.31}}

\newcommand{\AppHttAipwVr}{\ensuremath{1.136}}
\newcommand{\AppHttAipwVrMcse}{\ensuremath{0.089}}

\newcommand{\AppHttFacilities}{734}

\newcommand{\AppHttOverlapSens}{8.3}
\newcommand{\AppHttPsu}{183}

\newcommand{\AppHttRecords}{5{,}364}

\newcommand{\AppHttTmleRef}{\ensuremath{-2.96}}

\newcommand{\AppMinSampledAll}{24}

\newcommand{\AppMwAipwSeSd}{\ensuremath{1.089}}

\newcommand{\AppMwExcluded}{2}
\newcommand{\AppMwFacilities}{808}

\newcommand{\AppMwOverlapPrim}{0.0}
\newcommand{\AppMwOverlapSens}{35.6}
\newcommand{\AppMwPsu}{202}

\newcommand{\AppMwRecords}{6{,}894}
\newcommand{\AppMwSensTmleCov}{85.8}
\newcommand{\AppMwSensTmleVr}{\ensuremath{0.847}}
\newcommand{\AppMwTmleBoundZ}{2.3}

\newcommand{\AppMwTmleMeanSePp}{\ensuremath{5.12}}
\newcommand{\AppMwTmleRef}{\ensuremath{-0.10}}
\newcommand{\AppMwTmleSdPp}{\ensuremath{4.27}}

\newcommand{\AppMwTmleSeSd}{\ensuremath{1.199}}
\newcommand{\AppMwTmleVr}{\ensuremath{1.573}}
\newcommand{\AppMwTmleVrMcse}{\ensuremath{0.106}}
\newcommand{\AppSePctMax}{20}
\newcommand{\AppSePctMin}{4}
\newcommand{\SimBridge}{\ensuremath{0.063}}
\newcommand{\SimBridgeFirst}{\ensuremath{0.704}}
\newcommand{\SimBridgeInt}{[\ensuremath{0.053}, \ensuremath{0.085}]}
\newcommand{\SimBridgeLastTraj}{\ensuremath{0.261}}

\newcommand{\SimCells}{41}

\newcommand{\SimCfDqT}{\ensuremath{0.309}}
\newcommand{\SimCfDqTInt}{[\ensuremath{0.297}, \ensuremath{0.323}]}
\newcommand{\SimCfDqTwo}{\ensuremath{0.793}}
\newcommand{\SimCfDqTwoEight}{\ensuremath{0.888}}
\newcommand{\SimCfDqTwoFirst}{\ensuremath{1.326}}

\newcommand{\SimCfDqTwoSixteen}{\ensuremath{0.895}}

\newcommand{\SimCfDtT}{\ensuremath{0.279}}

\newcommand{\SimCfDtTInt}{[\ensuremath{0.268}, \ensuremath{0.289}]}
\newcommand{\SimCfDtTwo}{\ensuremath{0.578}}
\newcommand{\SimCfDtTwoDropPct}{57}
\newcommand{\SimCfDtTwoFirst}{\ensuremath{1.338}}
\newcommand{\SimCfDtTwoInt}{[\ensuremath{0.550}, \ensuremath{0.605}]}
\newcommand{\SimCfDtTwoTimes}{5.8}
\newcommand{\SimCfRaDropMax}{1.7}
\newcommand{\SimCfRaDropMin}{1.4}

\newcommand{\SimCfRaFirst}{\ensuremath{1.208}}
\newcommand{\SimCfRaQ}{\ensuremath{0.154}}
\newcommand{\SimCfRaQInt}{[\ensuremath{0.148}, \ensuremath{0.163}]}
\newcommand{\SimCfRaTimes}{1.5}
\newcommand{\SimCfRaTtwoTerm}{\ensuremath{0.075}}
\newcommand{\SimCoFpcKS}{\ensuremath{0.057}}
\newcommand{\SimCoFpcKSInt}{[\ensuremath{0.051}, \ensuremath{0.093}]}
\newcommand{\SimCoFpcMeanZ}{\ensuremath{-0.007}}
\newcommand{\SimCoFpcMeanZInt}{[\ensuremath{-0.097}, \ensuremath{0.081}]}

\newcommand{\SimCoFpcSDZ}{\ensuremath{1.203}}
\newcommand{\SimCoFpcSDZInt}{[\ensuremath{1.146}, \ensuremath{1.263}]}
\newcommand{\SimCoFpcVarExcessPct}{45}
\newcommand{\SimCoMsKS}{\ensuremath{0.020}}
\newcommand{\SimCoMsKSInt}{[\ensuremath{0.015}, \ensuremath{0.045}]}
\newcommand{\SimCoMsMeanZ}{\ensuremath{-0.015}}
\newcommand{\SimCoMsMeanZInt}{[\ensuremath{-0.059}, \ensuremath{0.029}]}
\newcommand{\SimCoMsSD}{\ensuremath{0.010}}
\newcommand{\SimCoMsSDInt}{[\ensuremath{-0.022}, \ensuremath{0.042}]}
\newcommand{\SimCoStrKS}{\ensuremath{0.031}}
\newcommand{\SimCoStrKSInt}{[\ensuremath{0.026}, \ensuremath{0.067}]}

\newcommand{\SimCoStrSD}{\ensuremath{0.069}}
\newcommand{\SimCoStrSDInt}{[\ensuremath{0.013}, \ensuremath{0.126}]}
\newcommand{\SimCovSEBig}{\ensuremath{0.005}}
\newcommand{\SimCovSESmall}{\ensuremath{0.008}}
\newcommand{\SimCsPsi}{\ensuremath{0.0066}}
\newcommand{\SimCsPsiInt}{[\ensuremath{0.0063}, \ensuremath{0.0072}]}
\newcommand{\SimCsSe}{\ensuremath{0.0013}}
\newcommand{\SimCsSeInt}{[\ensuremath{0.0012}, \ensuremath{0.0014}]}

\newcommand{\SimCtRaMean}{\ensuremath{0.008}}
\newcommand{\SimCtRaMeanInt}{[\ensuremath{0.006}, \ensuremath{0.011}]}

\newcommand{\SimCtRaSD}{\ensuremath{0.055}}

\newcommand{\SimCtRaTr}{\ensuremath{0.016}}

\newcommand{\SimCtRbMean}{\ensuremath{0.130}}
\newcommand{\SimCtRbMeanInt}{[\ensuremath{0.126}, \ensuremath{0.134}]}
\newcommand{\SimCtRbPct}{13.0}

\newcommand{\SimCtRbSEPct}{6.3}

\newcommand{\SimCtbJFourCovDetMax}{\ensuremath{0.993}}
\newcommand{\SimCtbJFourCovDetMed}{\ensuremath{0.986}}
\newcommand{\SimCtbJFourCovDetMin}{\ensuremath{0.978}}

\newcommand{\SimCtbJFourOverlayMed}{\ensuremath{0.985}}

\newcommand{\SimCtbWorstShortfall}{\ensuremath{0.006}}

\newcommand{\SimCthaMax}{\ensuremath{0.014}}
\newcommand{\SimCthaMin}{\ensuremath{-0.011}}
\newcommand{\SimCthaN}{12}

\newcommand{\SimEtCovNinetyFivePct}{96.9}

\newcommand{\SimEtOracleCovNinetyFivePct}{98.3}

\newcommand{\SimEtOracleSDZ}{\ensuremath{1.045}}
\newcommand{\SimEtRvar}{\ensuremath{1.146}}
\newcommand{\SimEtRvarPct}{14.6}

\newcommand{\SimEtRvarSqrtPct}{7}
\newcommand{\SimEtSigRatioMax}{\ensuremath{1.28}}

\newcommand{\SimEtSigRatioMin}{\ensuremath{1.22}}

\newcommand{\SimFsCovGapMax}{2.4}
\newcommand{\SimFsCovGapMin}{1.6}

\newcommand{\SimFsDqCovPct}{93.7}
\newcommand{\SimFsDqSDZ}{\ensuremath{1.141}}

\newcommand{\SimFsDtCovPct}{94.5}
\newcommand{\SimFsDtSDZ}{\ensuremath{1.073}}

\newcommand{\SimFsOrCovPct}{96.2}

\newcommand{\SimFsOrCovThreePct}{96.1}

\newcommand{\SimFsRaCovPct}{96.5}
\newcommand{\SimFsRaSDZ}{\ensuremath{0.982}}

\newcommand{\SimPcAlign}{\ensuremath{0.300}}
\newcommand{\SimPcAlignInt}{[\ensuremath{0.299}, \ensuremath{0.301}]}

\newcommand{\SimPcShift}{\ensuremath{-0.306}}
\newcommand{\SimPcShiftInt}{[\ensuremath{-0.372}, \ensuremath{-0.241}]}
\newcommand{\SimPcUnit}{\ensuremath{0.694}}
\newcommand{\SimPcUnitInt}{[\ensuremath{0.691}, \ensuremath{0.697}]}
\newcommand{\SimRbCovEightPct}{95.7}
\newcommand{\SimRbMeanZEight}{\ensuremath{0.26}}
\newcommand{\SimRbMeanZFifty}{\ensuremath{0.10}}

\newcommand{\SimRbUpper}{\ensuremath{-0.150}}

\newcommand{\SimStrictMax}{\ensuremath{0.41}}

\newcommand{\SimStrictMin}{\ensuremath{0.31}}

\newcommand{\SimUnresolvedOther}{33}
\newcommand{\SimUnresolvedTMLE}{47}
\newcommand{\SimUnresolvedTotal}{80}

\newcommand{\proofloc}{Section~\ref{S:supp:proofs} of the Supplementary Materials}

\title{Finite-population inference for cross-fitted estimators of treatment contrasts in stratified
multistage surveys}
\author{M.\ Ehsan Karim\\[4pt]
{\small School of Population and Public Health, University of British Columbia}\\
{\small Centre for Advancing Health Outcomes, St.\ Paul's Hospital}\\
{\small Vancouver, British Columbia, Canada}\\
{\small ehsan.karim@ubc.ca\quad ORCID 0000-0002-0346-2871}}
\date{}

\begin{document}

\maketitle

\noindent\textit{Running title:} Cross-fitted inference in stratified multistage surveys

\section*{Funding}
This research received no specific grant from any funding agency in the public, commercial, or not-for-profit
sectors.

\section*{Acknowledgements}
This research was supported in part through computational resources from Advanced Research Computing at the
University of British Columbia. The Service Provision Assessment data were obtained from the Demographic and Health Surveys (DHS) Program.
During the preparation of this manuscript, the author used Claude (Anthropic) to assist with the development of 
simulation and statistical analysis code, manuscript editing, and verification of mathematical derivations. All AI-assisted outputs were 
independently reviewed and verified by the author, who takes full responsibility for the accuracy and integrity of the work.

\section*{Conflict of interest}
The author declares no potential conflict of interest.

\clearpage

\begin{abstract}
\noindent Health surveys often sample clusters within strata and then individuals within clusters.
Fitting flexible outcome and treatment models complicates inference for adjusted contrasts from such surveys.
We establish finite-population limit theory for survey-weighted augmented inverse-probability-weighted estimation
within the double machine learning framework, using cross-fitting that keeps whole clusters together.
A cross-fitted targeted maximum likelihood estimator inherits the results under additional targeting conditions.
With first-stage clusters drawn by simple random sampling without replacement within strata, the theory allows
arbitrary within-cluster dependence, appreciable sampling fractions, and many small strata, a few large strata, or both.
Under bounded cluster sizes, a cluster-generation model and nuisance convergence-rate conditions, the estimator is
asymptotically normal about the finite-population contrast, with both the population and the sample treated as random.
A stability property derived from the cluster model controls the cross-fitting error without entropy conditions
on the nuisance-function classes.
We also show that the stratified ultimate-cluster variance estimator, computed from between-cluster variation
in weighted fitted-score totals without a first-stage finite-population correction, consistently estimates
a with-replacement variance.
This variance is at least as large as the without-replacement design variance governing the normal limit.
The resulting survey intervals for the finite-population contrast are asymptotically conservative and become
asymptotically exact when the excess is negligible relative to the with-replacement variance.
In simulations, the estimator with correctly specified nuisance models was close to its normal limit at the
largest sample size with many strata.
Nevertheless, near-nominal or higher coverage coexisted, in different settings, with remainder diagnostics above
the study's negligibility benchmark and with inaccurate standard errors.
We illustrate the methods with three Service Provision Assessment surveys from Malawi and Haiti, estimating
adjusted urban--rural differences in the proportion of respondents very satisfied with antenatal,
family-planning or sick-child care.
Across 500 stratified two-stage resamples per survey, with each set of eligible records treated as a fixed
population, mean reported standard errors were 4\%--20\% larger than the estimates' empirical standard deviations.
Coverage of reference values estimated from all eligible records in each survey frame was near or above
nominal despite this excess.
Comparing estimated and empirical variances across samples revealed the discrepancy.
\end{abstract}

\noindent\textit{Keywords:} Complement stability; Cross-fitting; Design-based inference; Double machine learning;
Finite population; Stratified multistage sampling; Targeted maximum likelihood estimation; Ultimate-cluster variance.

\section{Introduction}\label{sec:intro}

\paragraph{Survey setting and inferential target.}
Health-facility surveys such as the Service Provision Assessment (SPA) select facilities from a national list 
stratified by characteristics such as facility type and region. They then interview clients at the selected facilities
\citep{TurcotteTremblay2023}. \citet{TurcotteTremblay2023} used exit-interview data from five countries,
including Malawi and Haiti, to study factors associated with caretakers
reporting that they were very satisfied with sick-child care.
These factors included whether the facility was urban. Our illustration uses the SPA surveys of Malawi and Haiti to estimate
covariate-adjusted urban--rural differences in the proportion of clients
reporting that they are very satisfied
\citep{MalawiSPA2014,HaitiSPA2014,HaitiSPA2019}. Propensity-score methods with survey weights let such
adjusted estimates generalize to the survey's target population \citep{DuGoff2014}, and the weights can enter both
the propensity and the outcome models \citep{Ridgeway2015}.

Augmented inverse-probability-weighted (AIPW) estimation combines an outcome regression with a treatment (or
exposure) propensity model through a Neyman-orthogonal estimating score. When both models are fitted flexibly,
double machine learning (DML) pairs such scores with cross-fitting \citep{Chernozhukov2018}; the cross-fitted AIPW
estimator studied here is an instance. Targeted maximum likelihood estimation (TMLE) updates an initial outcome
regression through a targeting step \citep{vanderLaanRose2011}. Standard independent-observation arguments for
cross-fitting do not directly cover a stratified multistage design. Outcomes of clients in the same facility may be
dependent, and selecting facilities without replacement within a stratum makes their selection indicators
dependent. \citet{Chiang2022} develop cross-fitting for multiway clustered data; their sampling framework treats observations 
that share no cluster as independent, so it does not address the dependence created by selecting clusters without replacement.

The target also matters. Design-based inference holds the population fixed and repeats only the sampling; the
joint model--design framework also treats the population as one draw from a generating model
\citep{RubinBleuerKratina2005}. The finite-population contrast $\Psi_N$ is the population average of the AIPW
estimating function at the true nuisance functions; it describes the population
from which the survey was drawn. The super-population contrast $\Psi$ is the corresponding mean under the
generating model. Under the conditions used here, the difference between the two targets is negligible relative to the standard error
when the overall first-stage sampling fraction vanishes (Lemma~\ref{S:lem:modterm} in Section~\ref{S:supp:lem-modterm} of the
Supplementary Materials), as in the classical decomposition for a sample mean
\citep[equation~1.1]{RubinBleuerKratina2005}; this is a sufficient condition, not a necessary one. 
Facility surveys need not sample a small fraction of facilities, and
some set out to include them all \citep{HaitiSPA2019}. Our main result therefore concerns $\Psi_N$.

The sampling fraction also bears on the reported standard error. The usual stratified ultimate-cluster variance
estimator is computed from the variation of weighted totals between sampled primary sampling units (PSUs)
within each stratum \citep[Sections~2.4 and 6.15]{Wolter2007}. It treats the sampled PSUs as drawn with replacement, the setting of
classical linearization results \citep{KrewskiRao1981}, and so omits the first-stage finite-population
correction; here it is applied to the fitted score.

\paragraph{Related work.}
Design-based central limit theorems and variance estimators are classical for fixed statistics: under rejective
sampling with unequal probabilities \citep{Hajek1964}, stratified sampling \citep{BickelFreedman1984} and stratified
multistage sampling with PSUs drawn with replacement \citep{KrewskiRao1981}. Functional limits for Horvitz--Thompson
and H\'ajek estimators cover high-entropy single-stage designs \citep{Boistard2017}. For cross-fitted AIPW and TMLE,
these fixed-function results must be supplemented by control of the error introduced by estimating the nuisance
functions.

Recent work brings cross-fitting into survey sampling. Three works by other authors are closest to ours, and an
earlier preprint by the author addresses a related question. \citet{An2026} study
cross-fitted model-assisted estimation of a finite-population mean in the design-based framework, under single-stage
designs, with the population held fixed, auxiliary variables known for every population unit and expected sampling
fractions bounded away from zero. They characterize the designs under which a random fold partition of the population
makes the sampling indicators conditionally independent across folds (the exponential designs, which include simple
random, Bernoulli, Poisson and conditional Poisson sampling). \citet{DagdougHaziza2026} review these developments,
including a cross-fitted augmented estimator for item nonresponse whose learned response propensity requires a
product-rate condition, with a variance estimator for a negligible sampling fraction. \citet{ChibaMorikawa2026}
derive semiparametric efficiency bounds under the joint model--design law and give cross-fitted estimators of means
and regression coefficients that attain them. Their verified designs include stratified two-stage sampling with a
fixed number of strata and Poisson or rejective selection of PSUs, which includes simple random sampling without
replacement of PSUs, with each PSU kept whole within a fold and the folds fixed before sampling. Their framework
observes covariates for every population unit, allows no further dependence among outcomes within clusters and
estimates joint-law rather than design-conditional variances; they leave within-cluster outcome dependence and
design-conditional inference to future work. The earlier preprint \citep{KarimTMLE2026} examined cross-fitted
survey-weighted TMLE of a super-population average treatment effect under stratified multistage sampling, with PSUs
kept whole within folds. It studied asymptotic normality and consistency of the ultimate-cluster linearization
variance estimator under conditions that included a finite-population no-memorization condition on the out-of-fold
nuisance fits.

\paragraph{Contribution and scope.}
None of these works covers the combination studied here. Whole PSUs are selected without replacement within
strata, possibly two in each of a growing number of strata; folds are formed from the sampled PSUs; outcomes may be
arbitrarily dependent within a PSU; and only sampled units are observed. The target is a treatment contrast whose
propensity score must be learned, so, as for the item-nonresponse estimator, a product-rate condition on the
nuisance errors remains. The sampling fractions may vanish or stay appreciable, provided they are bounded away from
one. As in \citet{An2026}, the limit of the leading term is conditional on the finite population, whereas our
control of the fitted-nuisance errors uses a model for the clusters. Compared with the analysis in
\citet{KarimTMLE2026}, the primary
target is the finite-population contrast rather than a super-population effect, the stability property that
controls the cross-fitting error is derived from the cluster model rather than assumed, and the first-stage fraction
may stay appreciable, so the without-replacement design variance that governs the normal limit is distinguished from
the ultimate-cluster variance target.

Two difficulties arise in this setting. The estimator averages a \emph{fitted} influence function trained on the
same sampled clusters, so a classical design-based central limit theorem (CLT) for fixed functions does not apply
verbatim, and the usual route, a Donsker (entropy) condition on the nuisance class, is restrictive for flexible
learners. And when the first-stage fraction is appreciable, the ultimate-cluster variance estimator need not
estimate the without-replacement design variance that governs the normal approximation, so it is not obvious what
it does estimate.

We study the survey-weighted, PSU-cross-fitted AIPW estimator and its cross-fitted TMLE variant. Their ingredients
are established; our contribution is limit theory and a variance interpretation for the combination in the setting
just described. The main result is asymptotic normality about $\Psi_N$ under the joint model--design law,
standardized by the leading term's without-replacement design standard deviation. It holds under first-stage simple
random sampling without replacement (SRS-WOR), whether the first-stage fraction is appreciable or vanishing, with
many small strata, a few large ones, or both. Cross-fitting at the PSU level makes each held-out PSU's nuisances a
fixed function of the training folds. A property we call \emph{complement stability}, which we prove from a
stratified-mixture model for the clusters rather than assume, then controls the empirical-process term without an
entropy condition on the learners. The nuisance-rate conditions are still required. The leading term's limit is
design-based; the model for the clusters is used to control the fitted-nuisance errors (Section~\ref{sec:setup}).

We show that the usual stratified ultimate-cluster variance estimator, evaluated at the cross-fitted score,
consistently estimates the leading term's with-replacement variance target. Under SRS-WOR, this target is at least
as large as the leading term's without-replacement design variance, a classical comparison for fixed PSU totals
\citep[Corollary~2.4.3]{Wolter2007}. The resulting Wald
interval is asymptotically at least nominal. It is asymptotically exact when the relative gap between the two
variances vanishes. A corollary extends the result to cross-fitted TMLE. Extending the cross-fitting argument to
unequal-probability first stages, including the probability-proportional-to-size selection that many household
surveys use, is left open.

\paragraph{Practical implications and assessment.}
For an analyst, the result states when the familiar recipe of forming folds from whole PSUs within strata and
reporting the ultimate-cluster standard error is justified, and what that standard error estimates when the
first-stage fraction is not small. A simulation with known finite-population values (Section~\ref{sec:sim}) examines
how closely these asymptotic approximations held at the sample sizes studied, including where they did not, and a
resampling study on the records of the three SPA surveys under a known imposed design (Section~\ref{sec:app})
evaluates the implemented standard errors. Code for reproducing the reported numerical results from saved output is
supplied. 

Section~\ref{sec:setup} describes how the population is generated and the sample selected, the two
repeated-sampling frameworks, the target, the estimator and the conditions. Section~\ref{sec:main} states the
theorem, Section~\ref{sec:sim} reports the simulation, Section~\ref{sec:app} the real-data illustration, and
Section~\ref{sec:disc} interprets the findings. The supplement contains the proofs (S1), the simulation's methods
(S2) and additional results (S3), and the application's details (S4).

\section{Framework, target, and estimator}\label{sec:setup}

We distinguish the finite population from the model used to generate it, then define the sampling design, the
target, the cross-fitted estimator and the conditions. The distinction identifies which parts of the result are
conditional on the observed population and which require model assumptions. Figure~\ref{fig:targets} separates
population generation from survey selection. The primary estimation error is $\psihat-\Psi_N$; interpreting the
estimator relative to $\Psi$ adds the discrepancy $\Psi_N-\Psi=S_{\mathrm{mod}}$ of \eqref{eq:Ssplit}, whose
negligibility needs a separate argument. Lemmas~\ref{S:lem:designclt}--\ref{S:lem:crossfit} and the auxiliary Lemma~\ref{S:lem:envelope} are stated and proved in Section~\ref{S:supp:proofs} of the Supplementary Materials.

\begin{figure}[tbp]
\footnotesize
\centering
\includegraphics{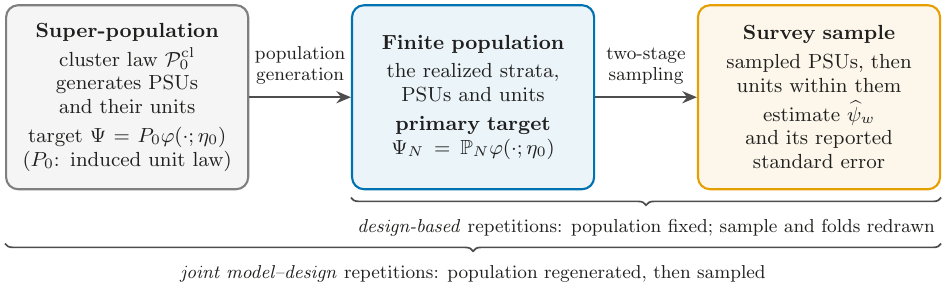}
\caption{Population generation, survey sampling and inferential targets. A super-population law $\Pzcl$ generates
a finite population of PSUs and their units. The super-population functional is $\Psi=\Pz\varphi(\cdot;\eta_0)$, with
$\Pz$ the size-weighted unit law induced by $\Pzcl$, and the realized population has census target
$\Psi_N=\PN\varphi(\cdot;\eta_0)$, the primary target of this paper. Two-stage survey sampling then produces $\psihat$
and its reported standard error. Design-based repetitions hold the finite population fixed and redraw the sample and
the fold assignment; joint model--design repetitions also regenerate the population. The fitted estimator's normal
limit (Theorem~\ref{thm:clt}) is established under the joint model--design law, centred at $\Psi_N$.}
\label{fig:targets}
\par\smallskip
{\scriptsize\raggedright\noindent\textbf{Alt text:} Three boxes from left to right, joined by arrows labelled
population generation and two-stage sampling: the super-population (cluster law and target $\Psi$), the finite
population (primary target $\Psi_N$) and the survey sample (estimate and its reported standard error). A brace under
the last two boxes marks design-based repetitions; a brace under all three marks joint model--design
repetitions.\par}
\end{figure}

\paragraph{Population generation.} A super-population law $\Pzcl$ generates a finite population of clusters
(PSUs) arranged in strata $h=1,\dots,H$, with $M_{\mathrm{pop},h}$ PSUs in stratum $h$ and
$M_{\mathrm{pop}}=\sum_hM_{\mathrm{pop},h}$ in all. Cluster $k=1,\dots,M_{\mathrm{pop}}$ carries a size $n_k$ and
units $\{O_{ki}\}_{i=1}^{n_k}$ that may be arbitrarily dependent within the cluster; the population holds
$N=\sum_k n_k$ units. The model for this step (Condition~\ref{cond:mix} below) draws clusters independently within
each stratum.

\paragraph{Sample selection.} A known two-stage scheme then selects PSUs and units. \emph{First stage:} $m_h\ge2$
PSUs are sampled without replacement per stratum, independently across strata, with known positive inclusion
probabilities $\pi_k$; $m=\sum_h m_h$ is the total number of sampled PSUs and $f_h:=m_h/M_{\mathrm{pop},h}$ the
first-stage sampling fraction of stratum $h$. Two summaries of these fractions are used and kept apart: the
\emph{largest stratum fraction} $\max_h f_h$ and the \emph{overall fraction} $m/M_{\mathrm{pop}}=\sum_h\rho_h f_h$,
$\rho_h=M_{\mathrm{pop},h}/M_{\mathrm{pop}}$. Because the overall fraction is a weighted average of the $f_h$,
$\max_h f_h\to0$ implies $m/M_{\mathrm{pop}}\to0$; the converse fails under unequal allocation. We retain general first-stage inclusion-probability notation $\pi_k$
because parts of the fixed-function design result (Lemma~\ref{S:lem:designclt} in Section~\ref{S:supp:lem-designclt} of the Supplementary Materials) hold beyond SRS-WOR: in the
many-strata regime for any independent, uniformly bounded strata, and the variance dominance also under
conditional-Poisson selection. The cross-fitted Theorem~\ref{thm:clt} assumes SRS-WOR at the PSU stage, while the
second-stage subsampling below may remain unequal-probability. \emph{Second stage:} within a sampled PSU
$k$, units are subsampled with known inclusion probabilities $\pi_{i\mid k}\ge\delta_2>0$, the selection in PSU $k$
generated locally from that PSU and a PSU-specific randomizer, independent across PSUs and of the first-stage and
fold randomization.

\paragraph{The resulting sample.} The full design weight of unit $i$ is
\[ w_i=1/(\pi_{k(i)}\pi_{i\mid k}). \]
Write $\Pn f=\Nhat^{-1}\sum_i w_i f(O_i)$, with $\Nhat=\sum_i w_i$, for the H\'ajek design average. The weight total
$\Nhat$ estimates the population size $N$, with $\Nhat/N\to1$ in probability (Lemma~\ref{S:lem:envelope} in Section~\ref{S:supp:lem-envelope} of the Supplementary Materials). This is
not in general an exact identity for a probability sample of PSUs: even under full within-PSU enumeration
($\pi_{i\mid k}\equiv1$) it is the first-stage estimator $\sum_{k\in s}n_k/\pi_k$, exact under a PSU census or a
constant PSU size. Finally, $\PN$ is the finite-population (census) mean and $\Pz$ the super-population mean.
Units within a sampled PSU remain dependent, PSUs sampled from the same stratum are dependent through
without-replacement selection, and strata are independent.

\paragraph{Standing conditions.} We impose five standing conditions:
\begin{enumerate}[label=(\roman*),leftmargin=*,itemsep=0pt,topsep=2pt,parsep=0pt]
\item bounded relative weights, $\max_i w_i\le C_w N/m$;
\item bounded cluster sizes, $1\le n_k\le n_{\max}$;
\item a bounded outcome, $|Y|\le C_Y$ a.s.;
\item first-stage sampling fractions $f_h$ bounded away from $1$;
\item the fitted outcome regression is clipped to a fixed $[q_0,1-q_0]\subset(0,1)$, and the fitted propensity is
truncated two-sidedly into $[\delta,1-\delta]$, with a fixed $\delta>0$ matching the positivity of $g_0$.
\end{enumerate}
Condition (iv) permits an appreciable fixed fraction; no vanishing-fraction rate is standing (the roles of a
vanishing fraction are set out under \emph{Regularity conditions} below). In (v), the clipping range is a
probability range for a binary $Y$, and both bounds are matched in the code. Under (iii), (v) and the standing
positivity below, the fitted influence function has a uniform envelope
$\|\varphi(\cdot;\widehat\eta)\|_\infty\le C_\varphi$: bounded $Y$, $\widehat Q\in[q_0,1-q_0]$ and truncated
$\widehat g\in[\delta,1-\delta]$ give $|H_g|\le\delta^{-1}$. The uniform score bound used throughout is therefore a
consequence of the estimator construction, not an assumption on the unclipped fit.

\paragraph{Repeated-sampling frameworks.} Two frameworks are used (Figure~\ref{fig:targets}), and they differ in what
is held fixed. In the
\emph{design-based} framework the finite population of clusters is fixed: its strata, PSUs, units, covariates,
treatments and outcomes. Each repetition redraws the first-stage PSU sample, the second-stage subsamples and the
fold assignment. Every design average and variance, including the target of the variance estimator
\eqref{eq:ucvar}, is computed over this randomness with the population held fixed. This framework determines the
distribution of the leading term $S_{\mathrm{des}}$ of the expansion \eqref{eq:onestep} below. In every regime
considered, the central limit theorem for $S_{\mathrm{des}}$ needs no model for how the population was generated.

In the \emph{joint model--design} framework the finite population is itself one draw from $\Pzcl$, regenerated at
each repetition before the sample is drawn. The fitted estimator's limit (Theorem~\ref{thm:clt}) holds in this joint
sense. It is centred at $\Psi_N$ and normed by the design variance, but the fitted-nuisance remainders $T_1,T_2$ are
controlled using the super-population layer: the fresh-complement step of Lemma~\ref{S:lem:crossfit} (Section~\ref{S:supp:lem-crossfit} of the Supplementary Materials) and the
$L^2(\Pz)$ rates. The further term $S_{\mathrm{mod}}$ (Lemma~\ref{S:lem:modterm}) is the separate bridge from the
census parameter to the super-population target $\Psi$. The simulation of Section~\ref{sec:sim} mirrors this split:
its oracle cells hold a few populations fixed and redraw many samples, and its fitted cells draw many populations
with a few samples each.

Condition~\ref{cond:he} collects the regularity used by the fixed-function design result.
In the many-strata regime, independent bounded strata yield the design CLT through the
across-strata Lindeberg--Feller theorem; in the fixed-strata and combined regimes,
the two-stage result is established here under SRS-WOR. No entropy condition is imposed on strata with bounded
$m_h$, including $m_h=2$. Broader fixed-function designs are discussed under that heading near the end of Section~\ref{S:supp:broader} of the Supplementary Materials.

\begin{assumption}[Fixed-function design regularity]\label{cond:he}
Second-order inclusion probabilities satisfy Sen--Yates--Grundy nonnegativity
$\pi_{kl}\le\pi_k\pi_l$ (standing). There are fixed constants $\bar F<\infty$ and $\bar m<\infty$, not depending
on $m$, such that for every $m$ at most $\bar F$ strata have $m_h>\bar m$ and all other strata have $m_h\le\bar m$.
The case in which no stratum has $m_h>\bar m$ is the \emph{many-strata} regime, the case $H\le\bar F$ the
\emph{fixed-strata} regime, and every other case the \emph{combined} regime (at most $\bar F$ larger strata
alongside strata with $m_h\le\bar m$).
\end{assumption}

\paragraph{Estimator and target.} The primary estimand is the \emph{finite-population} (design) center
$\Psi_N=\PN\varphi(\cdot;\eta_0)$, the census average of the AIPW estimating function at the true nuisances. The
\emph{super-population} functional $\Psi=\Pz\varphi(\cdot;\eta_0)$ is the size-weighted \emph{unit}-level mean under
the cluster-level law $\Pzcl$ (Condition~\ref{cond:mix}). It is recovered as the overall first-stage fraction
vanishes: $m/M_{\mathrm{pop}}\to0$ gives $\Psi_N-\Psi=\od_p(m^{-1/2})$ (Lemma~\ref{S:lem:modterm}). Both are identified
by an influence-function estimating function $\varphi$ with nuisance $\eta$. For the arm contrast (at $\Pz$, the
average-treatment-effect functional) with $O=(C,A,Y)$ and $\eta=(Q,g)$,
\begin{equation}\label{eq:phi}
\begin{aligned}
\Hg(A,C)&=\frac{A}{g(1\mid C)}-\frac{1-A}{1-g(1\mid C)},\\[2pt]
\varphi(O;\eta)&=\Hg(A,C)\{Y-Q(A,C)\}+Q(1,C)-Q(0,C),
\end{aligned}
\end{equation}
and $\D(O;\eta,\psi)=\varphi(O;\eta)-\psi$ is the estimating function, centred at $\psi=\Psi_N$ for the
design CLT and at $\psi=\Psi$ for the super-population reading. Here $Q(a,C)$ is the outcome regression and
$g(1\mid C)$ the propensity score. Their true values are $Q_0(a,C)=\E_{\Pz}(Y\mid A=a,C)$ and
$g_0(1\mid C)=\Pr_{\Pz}(A=1\mid C)$ under the size-weighted unit law $\Pz$ of Condition~\ref{cond:mix}. The two terms
of $\varphi$ are the inverse-probability-weighted residual and the outcome-model contrast; for $\Psi_N$ they give
\[
\Psi_N=\PN\{Q_0(1,C)-Q_0(0,C)\}+\PN[H_{g_0}(A,C)\{Y-Q_0(A,C)\}].
\]
This oracle-score census mean is the design center studied here; its residual component need not average to zero
in the realized population, and no identity with a realized finite-population average treatment effect is claimed.
At $\eta_0$, $\D$ is the
familiar AIPW score, the efficient influence function of the average treatment effect in the nonparametric model
for a single unit's data $O$. We use it only as an estimating function and make no efficiency claim under the survey
sampling model.

\noindent\textbf{Estimator.} The core estimator is the survey-weighted, cross-fitted \emph{augmented
inverse-probability-weighting / double-machine-learning} (AIPW/DML) plug-in of the influence function. Survey
weights enter the point estimate (Step~4), the PSU totals of the variance estimator (Step~5) and, for TMLE, the
targeting step:
\begin{enumerate}[leftmargin=*,label=\textbf{\arabic*.},itemsep=0pt,topsep=2pt,parsep=0pt]
\item \textbf{Split.} Partition whole PSUs into a \emph{fixed} number $V\ge2$ of folds (not depending on $m$) by a
uniform, data-independent balanced partition within each stratum, independent across strata, with equal
within-stratum shares $m_{v,h}=m_h/V$ ($V\mid m_h$, so the folds are exactly balanced;
Figure~\ref{fig:folds-variances}a); the
partition and the learner's internal randomness are exogenous to the data and to second-stage sampling (the fold
seed and held-out labels are not inputs to $\widehat\eta_{-v}$).
\item \textbf{Fit.} For fold $v$, estimate initial $Q$ and $g$ using only PSUs outside that fold, so
$\widehat\eta_{-v}$ is training-field measurable. This step does not prescribe whether the learner uses the
survey weights; the learners in the simulation (\S\ref{sec:sim}) are fitted without them and those in the
real-data illustration (\S\ref{sec:app}) with them.
\item \textbf{Predict.} Evaluate the nuisances on held-out PSUs and truncate $\widehat g$ two-sidedly to
$[\delta,1-\delta]$ and $\widehat Q$ to $[q_0,1-q_0]$ (condition (v)).
\item \textbf{Estimate.} Average the cross-fitted influence function with the survey weights,
$\psihat=\Pn\varphi(\cdot;\widehat\eta_{\mathrm{cf}})$, where on fold $v$ the nuisance is the held-out initial fit,
$\widehat\eta_{\mathrm{cf}}^{(v)}=\widehat\eta_{-v}$; equivalently $\Pn\D(\cdot;\widehat\eta_{\mathrm{cf}},\psihat)=0$.
\item \textbf{Infer.} Form the survey-weighted PSU-level totals of the centred fitted influence function,
$z_{hj}=\Nhat^{-1}\sum_{i\in(h,j)}w_i\{\varphi(O_i;\widehat\eta_{\mathrm{cf}})-\psihat\}$ for PSU $j$ of stratum $h$,
and compute the stratified ultimate-cluster variance
$\widehat{\mathrm{se}}_w^2=\sum_h\{m_h/(m_h-1)\}\sum_{j=1}^{m_h}(z_{hj}-\bar z_h)^2$, with $\bar z_h$ the stratum mean
of the $z_{hj}$. This is the ultimate-cluster estimator \eqref{eq:ucvar} below, computed without a
finite-population correction; Lemma~\ref{S:lem:designclt} establishes its properties.
\end{enumerate}
\emph{TMLE variant.} The cross-fitted TMLE replaces the held-out initial fit on each fold by one \emph{pooled},
low-dimensional ($\dim\le2$) \emph{logistic} targeting fluctuation of $Q$ on the logit scale,
$\widehat Q^\star=\mathrm{expit}\{\mathrm{logit}\,\widehat Q_{-v}+\widehat\epsilon^{\top}H\}$. The fluctuation is
common to all folds and holds the cross-fitted $g$ fixed. Its parameter $\widehat\epsilon$ solves the pooled
survey-weighted score equation $\Pn[H(Y-\widehat Q_\epsilon)]=0$, and $\widehat\eta_{\mathrm{cf}}^{(v)}$ is
$\widehat\eta_{-v}$ so fluctuated. Because the initial $\widehat Q_{-v}$ is clipped to $[q_0,1-q_0]$ (condition
(v)), the fluctuation keeps $\widehat Q^\star$ inside $(0,1)$ smoothly. No hard clip, which would zero the score
derivative, is applied after targeting; the uniform bound this needs is proved with Corollary~\ref{cor:tmle}
(Section~\ref{S:supp:outline-cor}). Theorem~\ref{thm:clt} is stated for the AIPW/DML plug-in, and the
complement-stability property below governs its out-of-fold \emph{initial} score $f_v$. The shared fluctuation,
absent for the plug-in ($\widehat\epsilon=0$), is not training-field measurable, so that property does \emph{not}
cover it. Corollary~\ref{cor:tmle} treats it as a finite-dimensional perturbation and shows that the cross-fitted
TMLE inherits the same limit under a small block of targeting conditions.

\paragraph{One-step expansion.} The proof rests on an exact decomposition of the estimation error. Adding and
subtracting $\varphi(\cdot;\eta_0)$ and population means in
$\psihat-\Psi_N=\Pn\varphi(\cdot;\widehat\eta_{\mathrm{cf}})-\PN\varphi(\cdot;\eta_0)$ gives the exact one-step
(von Mises) expansion about the design center,
\begin{equation}\label{eq:onestep}
\psihat-\Psi_N=\underbrace{(\Pn-\PN)\D(\cdot;\eta_0,\Psi_N)}_{S_{\mathrm{des}}}+T_1+T_2,
\end{equation}
with empirical-process term $T_1=(\Pn-\Pz)\{\varphi(\cdot;\widehat\eta_{\mathrm{cf}})-\varphi(\cdot;\eta_0)\}$ in
the fitted-minus-true nuisances and second-order remainder
$T_2=\Pz\{\varphi(\cdot;\widehat\eta_{\mathrm{cf}})-\varphi(\cdot;\eta_0)\}$, the latter's first-order nuisance
derivative vanishing by the Neyman orthogonality of $\varphi$ (Lemma~\ref{S:lem:prodrate} in Section~\ref{S:supp:lem-prodrate} of the Supplementary Materials). Because the fitted
function differs across folds, $\Pz$ of it means the fold-weighted sum
$\sum_v\omega_v\Pz\varphi(\cdot;\widehat\eta_{\mathrm{cf}}^{(v)})$, with $\omega_v=\widehat N_v/\Nhat$ the realized
share of the weight total in fold $v$ (not $1/V$); with this reading \eqref{eq:onestep} is exact. The leading
\emph{design} term $S_{\mathrm{des}}=(\Pn-\PN)\varphi(\cdot;\eta_0)$ has the difference operator annihilate the
constant ($\Pn\{1\}=\PN\{1\}=1$). The bridge to the super-population target is the single \emph{model} term,
\begin{equation}\label{eq:Ssplit}
\Psi_N-\Psi=S_{\mathrm{mod}}:=(\PN-\Pz)\D(\cdot;\eta_0,\Psi),\qquad
\psihat-\Psi=S_{\mathrm{des}}+S_{\mathrm{mod}}+T_1+T_2 .
\end{equation}
Thus the proof has four jobs: a design CLT for $S_{\mathrm{des}}$, cross-fitting control of $T_1$, a
product-rate bound for $T_2$, and --- only for the super-population reading --- negligibility of
$S_{\mathrm{mod}}$.

\paragraph{Two variance scales.} Two variance quantities for the leading term $S_{\mathrm{des}}$, more precisely
for its linearization (Lemma~\ref{S:lem:designclt}), are used throughout. The first, $\sigma_{m,\mathrm{wor}}^2$, is
its without-replacement design variance: the variance over repeated first- and second-stage sampling, including the
first-stage finite-population correction. It norms the theorem's normal limit. The second, $\sigma_m^2$, is the
with-replacement variance targeted by the ultimate-cluster variance estimator \eqref{eq:ucvar} of Step~5. That
estimator treats the sampled PSUs of a stratum as if they were drawn with replacement, and so omits the correction. Under SRS-WOR both split by stratum
into a between-PSU component $B_h$ and a within-PSU component $V_h$
(\proofloc; the fixed-total decomposition is classical, \citealp[Corollary~2.4.3]{Wolter2007}):
\[
\sigma_m^2=\sum_h(B_h+V_h),\qquad \sigma_{m,\mathrm{wor}}^2=\sum_h\{(1-f_h)B_h+V_h\}.
\]
Their gap $\Delta_m=\sigma_m^2-\sigma_{m,\mathrm{wor}}^2=\sum_h f_hB_h\ge0$ is the omitted correction. It is why the
usual interval can be conservative at an appreciable fraction. Whether that interval is asymptotically exact depends
on the \emph{relative} gap $\Delta_m/\sigma_m^2$ (Theorem~\ref{thm:clt}). Figure~\ref{fig:folds-variances}b shows that
the omitted correction acts only on the between-PSU components; the within-PSU components are the same on both
scales. This variance-calibration question is distinct from the target discrepancy $\Psi_N-\Psi$ of
Figure~\ref{fig:targets}.
For a fixed bounded function $f$, the ultimate-cluster estimator is
\begin{equation}\label{eq:ucvar}
\begin{gathered}
\widehat{\mathrm{se}}^2=\sum_{h}\frac{m_h}{m_h-1}\sum_{j=1}^{m_h}(z_{hj}-\bar z_{h})^2, \\
z_{hj}=\Nhat^{-1}\!\!\sum_{i\in(h,j)}\!\! w_i\{f(O_i)-\Pn f\},\quad \bar z_h=\tfrac1{m_h}\sum_{j} z_{hj},
\end{gathered}
\end{equation}
in which the observed (second-stage Horvitz--Thompson) weighted PSU totals $z_{hj}$ carry the within-PSU subsampling
variance. The target $\sigma_m^2$ is the design expectation of the census-referenced version of \eqref{eq:ucvar}, in
which the $z_{hj}$ use the census $N$ and $\PN f$ in place of $\Nhat$ and $\Pn f$ (\proofloc). Under the
fixed-function design conditions of Lemma~\ref{S:lem:designclt}, the reported $\widehat{\mathrm{se}}^2$ differs from
that version by $\od_p(\sigma_m^2)$.

\begin{figure}[tbp]
\footnotesize
\centering
\includegraphics{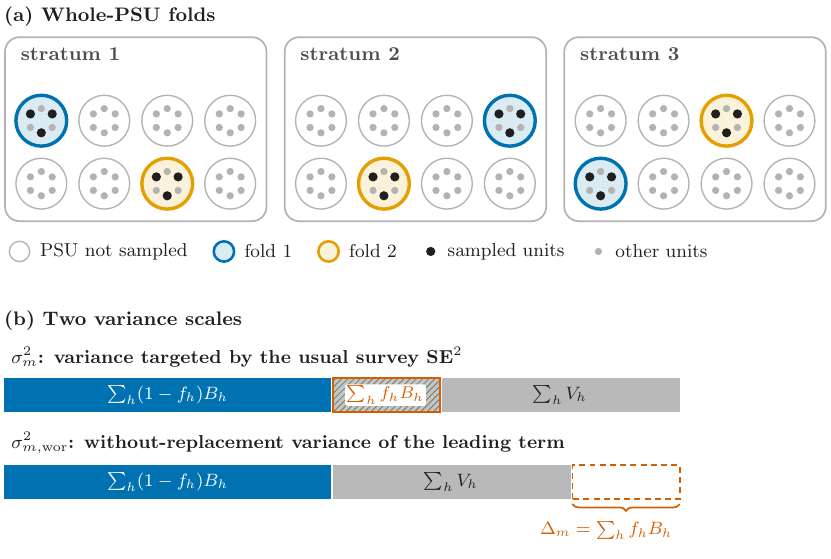}
\caption{Whole-PSU cross-fitting and the two variance scales. (a) Three strata are illustrated, each containing
eight PSUs, of which two are sampled without replacement. Sampled PSUs are assigned intact to two folds, balanced
within strata; dark dots denote sampled units and grey dots other units. The initial nuisance models for each fold
are fitted on the sampled PSUs of the other fold and evaluated on the held-out fold, and the roles are then
exchanged; the TMLE's subsequent targeting fluctuation is pooled across folds. Two folds are illustrated;
Theorem~\ref{thm:clt} allows any fixed number $V\ge2$ with $V\mid m_h$. (b) For the linearized leading design term,
$\sigma_m^2=\sum_h(B_h+V_h)$ is the target of the ultimate-cluster variance estimator \eqref{eq:ucvar}, whereas
$\swor^2=\sum_h\{(1-f_h)B_h+V_h\}$ is the without-replacement design variance. Their difference,
$\Delta_m=\sum_hf_hB_h$, affects only the between-PSU components; the within-PSU components are unchanged. Bar
lengths are schematic, drawn with $f_h=0.25$.}
\label{fig:folds-variances}
\par\smallskip
{\scriptsize\raggedright\noindent\textbf{Alt text:} (a) Three strata of eight PSUs each, drawn as circles holding six
units. In each stratum two PSUs are sampled, one in fold 1 and one in fold 2, and three of the six units of each
sampled PSU are sampled. (b) Two horizontal bars. The upper bar, $\sigma_m^2$, has a between-PSU part
$\sum_h(1-f_h)B_h$, a hatched part $\sum_hf_hB_h$ and a within-PSU part $\sum_hV_h$; the lower bar, $\swor^2$, lacks
the hatched part and is shorter by $\Delta_m$.\par}
\end{figure}

\paragraph{Regularity conditions.} Let $f_v=\varphi(\cdot;\widehat\eta_{-v})-\varphi(\cdot;\eta_0)$ be the
\emph{single realized} fold-$v$ score, where $\widehat\eta_{-v}$ is the out-of-fold \emph{initial} nuisance fit.
Its error is bounded by a deterministic sequence $r_m\to0$: $\|\widehat\eta_{-v}-\eta_0\|_{L^2(\Pz)}\le r_m$ with
probability tending to one, where
$\|\widehat\eta_{-v}-\eta_0\|_{L^2(\Pz)}=\max\{\|\widehat Q_{-v}-Q_0\|_{\Pz},\|\widehat g_{-v}-g_0\|_{\Pz}\}$ and
the outcome norm is joint over both arms, as in Lemma~\ref{S:lem:prodrate}. The sequence $r_m$ bounds the realized
out-of-fold error; it is not that error itself.

Two rate conditions are used. Lemma~\ref{S:lem:crossfit}, for $T_1$, needs only $r_m\to0$.
Lemma~\ref{S:lem:prodrate} bounds $T_2$ by the product of the outcome and propensity errors.
Theorem~\ref{thm:clt} and Corollary~\ref{cor:tmle} assume that each error is $\od_p(m^{-1/4})$,
a sufficient condition for their product to be $\od_p(m^{-1/2})$.
Under these separate rates, $r_m$ can be chosen with $r_m=\od(m^{-1/4})$, as the proof of
Corollary~\ref{cor:tmle} uses (\proofloc\ explains this choice).

No first-stage-fraction rate is standing. The design CLT about $\Psi_N$ and the cross-fitting control of $T_1$
(Lemma~\ref{S:lem:crossfit}) hold whether the fraction is appreciable or vanishing. A vanishing overall fraction
$m/M_{\mathrm{pop}}\to0$ enters in two places, by \emph{distinct} mechanisms. For the \emph{super-population} target
$\Psi$, it makes $S_{\mathrm{mod}}$ negligible (Lemma~\ref{S:lem:modterm}, Condition~\ref{cond:nd}). The norming
calibration works differently: it holds \emph{iff} $\Delta_m/\sigma_m^2\to0$. For that, $\max_h f_h\to0$ and, under
bounded PSU totals, $m/M_{\mathrm{pop}}\to0$ are each \emph{sufficient but not necessary}, and $S_{\mathrm{mod}}$ plays
no part.

\paragraph{Complement stability (intermediate property).} Cross-fitting works because each held-out PSU's score
is a fixed function of the other folds. What must be shown is that, on the PSUs the learner did not see, this
fixed function behaves as it would under the super-population law; complement stability states this precisely.
For a fold-$v$ score $g$ let $\theta_v(g)$ be the
\emph{size-weighted poststratified complement mean}
\[
\theta_v(g)=\frac{\sum_h\rho_h\,\bar t_h^{(-v)}(g)}{\sum_h\rho_h\,\bar n_h^{(-v)}},\qquad
\rho_h=\frac{M_{\mathrm{pop},h}}{M_{\mathrm{pop}}},
\]
the population-weighted \emph{ratio} of the complement-frame averages $\bar t_h^{(-v)}(g)$ and $\bar n_h^{(-v)}$
of the cluster totals $\sum_{i\in k}g(O_{ki})$ and sizes $n_k$, taken over each stratum's frame minus the
fold-$v$ training PSUs. It is the finite-population complement analogue of the size-weighted unit law $\Pz$
(Condition~\ref{cond:mix}), reducing to a per-PSU average of cluster totals only when cluster sizes are
non-informative. The cross-fitting argument turns on the \emph{complement-stability} property: conditionally on
the fold-$v$ training field $\mathcal T_v$, with probability $\to1$ the realized initial score $f_v$ obeys
\begin{equation}\label{eq:cs}
|\theta_v(f_v)-\Pz f_v|=\od_p(m^{-1/2}),\qquad
\theta_v(f_v^2)=O_p(r_m^2)+\od_p(m^{-1/2}).
\end{equation}
Stating the property on the \emph{complement} frame, where $f_v$ is out of sample, is what removes any
first-stage-fraction rate. A memorizing or interpolating learner can inflate the full-census $\PN f_v^2$ and
$(\PN-\Pz)f_v$ to order $\max_h f_h$ through its own training PSUs, but the complement excludes those PSUs.
\emph{This property is not assumed:} Lemma~\ref{S:lem:crossfit} \emph{proves} \eqref{eq:cs}, in its fresh-complement
step, from the stratified-mixture model (Condition~\ref{cond:mix}), training-only whole-PSU cross-fitting, the
bounded envelope and the out-of-fold $L^2(\Pz)$ rate $r_m\to0$. It constrains \emph{one} fixed
$\mathcal T_v$-measurable function, not a class over the nuisance model, so it is not a Donsker or entropy
condition. The CLT therefore covers cross-fitted learners for which such a condition cannot be verified, and none is
imposed.

Condition~\ref{cond:nd} rules out designs whose leading design variance is not of order $1/m$.

\begin{assumption}[Design nondegeneracy]\label{cond:nd}
$m\sigma_{m,\mathrm{wor}}^2\to\bar\sigma^2\in(0,\infty)$, where $\sigma_{m,\mathrm{wor}}^2$ is the
without-replacement design variance of the \emph{linearized} leading term $L$ of
Lemma~\ref{S:lem:designclt} (the H\'ajek difference $S_{\mathrm{des}}$ differs from $L$ by $O_p(m^{-1})$
and carries the same norming); the ultimate-cluster
(with-replacement) variance target $\sigma_m^2$ of Lemma~\ref{S:lem:designclt} has $m\sigma_m^2$ likewise bounded away
from $0$ and $\infty$
(equal to $\bar\sigma^2$ as the relative gap $\sum_h f_h B_h/\sigma_m^2\to0$, e.g.\ as $\max_h f_h\to0$), for
$\Pzcl$-almost every population sequence.
\end{assumption}

Condition~\ref{cond:mix} is the model layer. It connects the finite-population center to the super-population
target, and it controls the fitted-nuisance remainders $T_1,T_2$ through the fresh-complement step of
Lemma~\ref{S:lem:crossfit} and the $L^2(\Pz)$ product rate. The leading design CLT does not use it in any regime. In
words, clusters within a stratum are independent draws from a common stratum-specific law.

\begin{assumption}[Stratified-mixture population / aggregate-centering]\label{cond:mix}
The population clusters are drawn independently, identically distributed \emph{within} each stratum $h$ from a
stratum \emph{cluster}-level law $P_h$, and the super-population cluster law is the stratified mixture
$\Pzcl=\sum_h\rho_h P_h$ at the design stratum proportions $\rho_h=M_{\mathrm{pop},h}/M_{\mathrm{pop}}$. With
$\mu_h(f)=\E_{P_h}\{\sum_{i\in k}f(O_{ki})\}$ the expected within-cluster total of $f$ under $P_h$, the induced
size-weighted \emph{unit} law is $\Pz f=\{\sum_h\rho_h\mu_h(f)\}/\{\sum_h\rho_h\mu_h(1)\}$ (size-biased under
informative cluster sizes), so the unit-level target is $\Psi=\Pz\varphi(\cdot;\eta_0)$ and the $L^2(\Pz)$ norms
of the complement-stability property \eqref{eq:cs}, Condition~\ref{cond:nd}, and Lemma~\ref{S:lem:prodrate} are
taken under it. A function is
aggregate-centered exactly when $\sum_h\rho_h\mu_h(g)=0$, in which case $\Pz g=0$; in particular
$g=\D(\cdot;\eta_0,\Psi)$ is aggregate-centered.
\end{assumption}

The proportions $\rho_h$, and in the many-strata regime the strata themselves, change with $m$, so $\Pzcl$, $\Pz$,
$\eta_0$ and $\Psi$ are indexed by $m$; we suppress the index. Aggregate centering is imposed at the finite-$m$
proportions $\rho_h$, as the proof of Lemma~\ref{S:lem:modterm} requires.

\paragraph{Scope in brief.} The restrictions that matter for use are these. The first stage is SRS-WOR within
strata (Theorem~\ref{thm:clt}; parts of the fixed-function design result are broader). The first-stage fraction
may be appreciable; the ultimate-cluster interval is then asymptotically conservative. At appreciable fractions we
retain the finite-population target $\Psi_N$; a vanishing overall fraction, $m/M_{\mathrm{pop}}\to0$, is sufficient
to replace it by $\Psi$. Folds are a fixed number $V$ of whole-PSU groups balanced within each stratum
($V\mid m_h$). The out-of-fold nuisance errors must vanish in $L^2(\Pz)$ for $T_1$.
For $T_2$, Lemma~\ref{S:lem:prodrate} uses their product; the separate $\od_p(m^{-1/4})$ rates
assumed here are a convenient sufficient condition. Positivity holds with the clipping and
truncation of (v). The leading design CLT needs no model in any regime. The stratified-mixture model is used for
the nuisance errors and the bridge to $\Psi$, so the limit of the fitted estimator is a joint model--design
statement.

\section{Main results}\label{sec:main}

The theorem has three parts: a normal limit for the estimator about $\Psi_N$, consistency of the ultimate-cluster
variance estimator for its with-replacement target $\sigma_m^2$, and the relation between the two normings, which
decides whether the resulting interval is conservative or exact.

\begin{theorem}[Cross-fitted cluster CLT and conservative inference]\label{thm:clt}
Assume the conditions of \S\ref{sec:setup}, Conditions~\ref{cond:he}--\ref{cond:mix}, \emph{first-stage PSU
selection by simple random sampling without replacement} (allocation across strata satisfying the standing
conditions; first-stage fractions bounded away from one), PSU-level cross-fitting on a \emph{fixed} number $V\ge2$
of whole-PSU folds (not depending on $m$) with equal within-stratum shares $m_{v,h}=m_h/V$ ($V\mid m_h$), and the
separate nuisance rates of Lemma~\ref{S:lem:prodrate} (each out-of-fold outcome and propensity error
$\od_p(m^{-1/4})$). Then, for the cross-fitted
AIPW/DML plug-in $\psihat=\Pn\varphi(\cdot;\widehat\eta_{\mathrm{cf}})$ (Step~4, $\widehat\eta_{\mathrm{cf}}^{(v)}=
\widehat\eta_{-v}$, no targeting) with influence function $\D$:

\noindent\emph{Point CLT (design center).}
\[
\frac{\psihat-\Psi_N}{\sigma_{m,\mathrm{wor}}}\xrightarrow{d}\mathcal N(0,1),
\qquad \Psi_N=\PN\varphi(\cdot;\eta_0),
\]
under these conditions; if further $m/M_{\mathrm{pop}}\to0$ then $S_{\mathrm{mod}}=\Psi_N-\Psi=
\od_p(\sigma_{m,\mathrm{wor}})$ (Lemma~\ref{S:lem:modterm}) and the same holds about the super-population $\Psi$.

\noindent\emph{Variance estimation.} With $\widehat{\mathrm{se}}_w^2$ the estimator \eqref{eq:ucvar} computed on
$\widehat\D_i=\D(O_i;\widehat\eta_{\mathrm{cf}},\psihat)$, $\widehat{\mathrm{se}}_w^2/\sigma_m^2\xrightarrow{p}1$.

\noindent\emph{Conservativeness and calibration.} Under SRS-WOR, $\sigma_{m,\mathrm{wor}}^2\le\sigma_m^2$ exactly
(Lemma~\ref{S:lem:designclt}), with gap $\Delta_m:=\sigma_m^2-\sigma_{m,\mathrm{wor}}^2=\sum_h f_h B_h\ge0$, where
$B_h$ is the between-PSU component of the stratum-$h$ variance target. The Wald interval is \emph{asymptotically at
least nominal}: for each nominal level $1-\gamma$,
$\liminf_m\Pr(\Psi_N\in\psihat\pm z_{1-\gamma/2}\widehat{\mathrm{se}}_w)\ge1-\gamma$ under the joint law. The
normings calibrate, $\sigma_{m,\mathrm{wor}}/\sigma_m\to1$, iff the \emph{relative} gap $\Delta_m/\sigma_m^2\to0$, in
probability under the joint law, since the gap depends on the generated population. Under this condition
$(\psihat-\Psi_N)/\widehat{\mathrm{se}}_w\xrightarrow{d}\mathcal N(0,1)$, so the ultimate-cluster interval is
asymptotically exact about the \emph{design center}. If $m/M_{\mathrm{pop}}\to0$, then $S_{\mathrm{mod}}$ is
negligible and, because PSU totals are bounded under the standing conditions, the normings calibrate; hence
$(\psihat-\Psi)/\widehat{\mathrm{se}}_w\xrightarrow{d}\mathcal N(0,1)$ about the super-population target as well.
\end{theorem}

\noindent\textit{Proof.} Section~\ref{S:supp:outline-thm} of the Supplementary Materials assembles the proof from
Lemmas~\ref{S:lem:designclt}--\ref{S:lem:crossfit}, which Section~\ref{S:supp:proofs} states and proves.

The distributional limit is under the \emph{joint model--design} law. The design, re-drawing PSUs from the fixed
finite population, delivers the $S_{\mathrm{des}}$ central limit theorem free of any cluster-versus-unit model. The
nuisance-error terms $T_1,T_2$ are made $\od_p(\sigma_{m,\mathrm{wor}})$ using the super-population $L^2(\Pz)$
consistency ($r_m\to0$) and Lemma~\ref{S:lem:prodrate}, that is, through the model layer of
Condition~\ref{cond:mix}.

\noindent\emph{What this means for an analyst.} Under SRS-WOR of PSUs, the ultimate-cluster standard error,
computed from whole-PSU folds and without a finite-population correction, gives an interval that is asymptotically
at least nominal for $\Psi_N$. Its excess coverage is governed by the relative gap $\Delta_m/\sigma_m^2$, that is,
by the first-stage fractions $f_h$ weighted by the between-PSU variance components $B_h$. When the overall fraction
$m/M_{\mathrm{pop}}$ vanishes (with bounded PSU totals), the same interval is asymptotically exact for both $\Psi_N$
and $\Psi$; this holds in particular when all stratum fractions vanish.
Across the simulated populations and fractions, the mean gap $\Delta_m$ was close to the
across-population variance of $\Psi_N-\Psi$ (Table~\ref{S:tab:S-vanish}).
This suggests that the extra interval width that makes coverage conservative for $\Psi_N$
is roughly the additional width an interval for $\Psi$ would need.
The paper does not derive this relation or establish fixed-fraction coverage about $\Psi$.
The theorem does not say how large $m$ must be for the remainders to be negligible; Section~\ref{sec:sim} examines
that question empirically.

\begin{corollary}[Cross-fitted survey-weighted TMLE]\label{cor:tmle}
Consider the \emph{TMLE variant} of \S\ref{sec:setup}: the held-out initial fit $\widehat\eta_{-v}$ composed with
a single \emph{pooled} logistic fluctuation $\widehat Q^\star=\mathrm{expit}\{\mathrm{logit}\,\widehat Q_{-v}
+\widehat\epsilon^{\top}H\}$ common to all folds, holding the cross-fitted $\widehat g$ fixed. Assume, beyond the
hypotheses of Theorem~\ref{thm:clt}, the \emph{targeting conditions}: \emph{(T1)} the clever covariate $H$, built
on each fold from that fold's propensity $\widehat g_{-v}$, is bounded and the fluctuation dimension is fixed
($\dim\le2$); \emph{(T2)} the pooled weighted Gram matrix $G_n=\Pn[HH^{\top}]$, with each unit evaluated at its own
fold's $H$, has $\lambda_{\min}(G_n)\ge c_H>0$ with probability $\to1$; \emph{(T3)} the fluctuation is
the smooth logistic tilt above of the \emph{clipped} initial $\widehat Q_{-v}\in[q_0,1-q_0]$ (condition (v)), with
no post-fluctuation hard clip. Then the cross-fitted TMLE $\widehat\psi_{\mathrm{TMLE}}=\Pn\varphi(\cdot;
\widehat\eta_{\mathrm{cf}})$ satisfies the same Point CLT, variance-consistency, conservativeness, and calibration
as in Theorem~\ref{thm:clt}.
\end{corollary}

\noindent\textit{Proof.} Section~\ref{S:supp:outline-cor} of the Supplementary Materials gives the proof.

\paragraph{Structure of the proof.} The proof separates the four terms of \eqref{eq:onestep}--\eqref{eq:Ssplit}.
Figure~\ref{S:fig:proof-map} of the Supplementary Materials summarizes the four terms and the arguments that control
them.
The fixed-function design result (Lemma~\ref{S:lem:designclt}) gives the normal limit of $S_{\mathrm{des}}$.
It also establishes consistency of the ultimate-cluster variance estimator.
Complement stability, proved in Lemma~\ref{S:lem:crossfit}, controls $T_1$ under the stratified-mixture model.
The product of the nuisance errors controls $T_2$ (Lemma~\ref{S:lem:prodrate}).
Only the super-population interpretation also requires $S_{\mathrm{mod}}$ to be negligible (Lemma~\ref{S:lem:modterm}).
Section~\ref{S:supp:proofs} of the Supplementary Materials states and proves the four lemmas.

\section{Simulation study}\label{sec:sim}

The theorem is asymptotic. The simulation asks whether its approximations were useful at the sample sizes studied,
and whether the estimated standard errors tracked their theoretical target. We report the bias, dispersion and
distance from normality of the standardized error, the accuracy of the variance estimator, the size of the
fitted-nuisance remainders and interval coverage, each with its Monte Carlo uncertainty. This section summarizes the
design and the main results. The supplement gives the full design and numerical methods (Section~\ref{S:supp:sim})
and further results (Section~\ref{S:supp:simres}); the complete configuration-level results are in the reproduction 
repository \citep{KarimCode2026}.

\subsection{Design}\label{sec:sim-design}

\emph{Population generation.} Each finite population is a set of strata of PSUs. A PSU-level covariate
$C_{\mathrm{psu}}\sim U(-1,1)$ drives the PSU size, $n_k=34+14C_{\mathrm{psu}}$ rounded to $\{20,\dots,48\}$.
Units carry four covariates $C_1,\dots,C_4\sim U(-1,1)$, and a class label $K_B\in\{1,2\}$, alternating across
strata, sets the variances of PSU random intercepts in the treatment and outcome models. Treatment $A$ and binary
outcome $Y$ follow logistic models in these covariates with a $C_1C_3$ interaction, a class-specific interaction
and a treatment-by-$C_{\mathrm{psu}}$ term, so neither nuisance function is additive and neither remainder is
trivially zero. The nuisance truth $\eta_0$ is the model marginalized over the random intercepts. The targets
$\Psi_N$ and $\Psi$, the remainders $T_1,T_2$ and the normings $\swor$ and $\sigma_m$ are computed for each
population from its census and by numerical integration, so the truth is known. Clusters are drawn independently
within strata from this law, which satisfies Condition~\ref{cond:mix} by construction.

\emph{Sample selection.} From each population, $m_h$ PSUs per stratum are drawn by SRS-WOR, and $n_0=20$ units are
drawn by SRS-WOR within each selected PSU, so $\pi_{i\mid k}=20/n_k$. Two regimes are indexed by the total number
of sampled PSUs $m$. In the \emph{many-strata} regime $m_h=2$ and $H$ grows, with $m$ from 50 to 1600 at
first-stage fraction $f=0.10$; it is used for the normal approximation, the variance estimator and the TMLE
corollary. In the \emph{fixed-strata} regime $H=5$ and $m_h$ grows at $f=0.10$, with $m$ from 50 to 3200; it is used
for the remainders. Further constructions examine the without-replacement gap and the bridge to $\Psi$.
A fraction sweep fixes $H=50$ and $m_h=2$ ($m=100$) and varies
$f\in\{0.02,0.10,0.25,0.50\}$; AIPW with the correctly specified learner (R-a) and with a random forest (R-b), 
both described below, is also fitted at $f=0.50$.
A fixed-fraction trajectory holds $m_h=2$ and $f=0.50$ while $H$ grows from 25 to 400;
its relative gap is computed exactly from census totals, with R-a and R-b AIPW fitted at $m=200$.
A census-only vanishing-fraction trajectory examines the bridge to $\Psi$.
Section~\ref{S:supp:experiments} of the Supplementary Materials describes the implementation checks.

\emph{Estimators and learners.} Unless otherwise specified, fitted arms use the procedure of Section~\ref{sec:setup} with $V=2$ whole-PSU
folds balanced within strata. Survey weights enter the point estimate, the variance and the TMLE targeting step.
The learners are fitted without weights, which is appropriate here because PSU size depends only on the observed
$C_{\mathrm{psu}}$, so second-stage selection is ignorable given the covariates. The arms differ in the nuisance
fit. An \emph{oracle} arm uses $\eta_0$, so its error is the design term $S_{\mathrm{des}}$ alone. A correctly specified 
random-intercept logistic learner (R-a) is the parametric case: errors of order $m^{-1/2}$, which would meet the rate
conditions, are expected under appropriate regularity conditions, but correct specification alone does not guarantee
them, and no pilot slope diagnostic was run for R-a.
A random forest (R-b, many-strata regime
and the fold-number sensitivity) is the flexible case. A throttled learner (R-d, fixed-strata regime only) refits R-a on a 
random subset of
$\lfloor 30(m/200)^{2\alpha}\rfloor$ training PSUs, designed so that its error decays at the rate $m^{-\alpha}$. Four
rates were designed: $\alpha=0.20$ (below the quarter-rate threshold of Lemma~\ref{S:lem:prodrate}), $0.25$ (on it),
and $0.30$ and $0.35$ (above it). A rate check on pilot replicates found the fitted error slopes compatible with
the designed rates at $\alpha=0.25$ and $0.30$; at $\alpha=0.20$ and $0.35$ the errors decayed faster than designed,
and those two schedules were not run (Table~\ref{S:tab:S-rates}). R-a and R-b are also run as cross-fitted TMLE.
Fitted cells use 400 populations with 5 samples each (160 populations with 5 samples at $f=0.50$
and in the control conditions); oracle cells use 12 fixed populations with 4,000 samples each
(2,000 in the fraction sweep; 10,000 at $f=0.50$ for the variance-gap check).

\subsection{Performance measures, benchmarks and Monte Carlo uncertainty}\label{sec:sim-grading}

The measures are standardized by the theorem's normings, and each except the Kolmogorov distance has a benchmark set
before the main runs (Table~\ref{S:tab:S-margins}). For the normal approximation, the benchmarks
are that $Z=(\hat\psi-\Psi_N)/\swor$ have a mean within $\pm0.10$ and a standard deviation within $1\pm0.05$; its
Kolmogorov distance from the standard normal is reported beside 0.05 as a descriptive reference value. For the
variance estimator, the benchmarks are that
$R_{\mathrm{var}}=\widehat{\mathrm{se}}^2/\sigma_m^2$ have a mean within $1\pm0.15$ and a standard deviation of at
most 0.15, and that its mean absolute difference from the oracle's variance estimate on the same sample, relative to
$\sigma_m^2$ (the transfer), be at most 0.25. For the remainders and the bridge to $\Psi$, the benchmark is that the
0.9 quantile of the absolute term divided by $\swor$ be at most 0.10: nine in ten replicates should have a remainder
below a tenth of $\swor$, the without-replacement design standard deviation of the leading term, which is smaller
than the ultimate-cluster standard error when the fraction is appreciable. We use 0.10 as a descriptive reference
for the magnitude of these standardized terms. It is an author-selected benchmark, not a theorem-derived threshold
or a guarantee about coverage; the other benchmarks are likewise the study's own choices, not general standards.

Three variability quantities appear throughout and are kept distinct. Two are the variance scales of
Section~\ref{sec:setup}: $\swor^2$, the norming of the theorem's normal limit, and $\sigma_m^2$, the quantity the
squared standard error estimates, which under SRS-WOR is at least $\swor^2$. The third is the empirical spread of
$\hat\psi$ about $\Psi_N$ across replicates, summarized by SD$(Z)$ in units of $\swor$; it includes in addition the
remainders $T_1$ and $T_2$.

Two kinds of interval also appear, and we name them differently. A \emph{confidence interval} is what an analyst
would report, $\hat\psi\pm z_{1-\gamma/2}\widehat{\mathrm{se}}$; its coverage of $\Psi_N$ is a simulation result. A
\emph{Monte Carlo interval} expresses the simulation's own uncertainty about a summary such as SD$(Z)$. It is the
95\% percentile interval of 10,000 bootstrap resamples of the saved output, resampling whole populations for fitted
cells (each with its five samples, so that samples from the same population are not treated as independent), samples
for oracle cells and populations for census cells. For the Kolmogorov distance, the estimate and every resampled
value include empirical-distribution noise, so its Monte Carlo interval summarizes resampling variability but is not
treated as a calibrated confidence interval for the distance, particularly near zero; its comparisons with 0.05 are
descriptive.

Each benchmarked measure is compared with its benchmark through its Monte Carlo interval: the benchmark is met if the interval
lies inside the benchmark region, not met if it lies wholly outside, and the comparison is unresolved at the
simulated precision if the interval crosses the boundary. Comparisons are made at each arm's largest sample size;
smaller sizes are reported as trajectories. Three kinds of result are reported descriptively beside the benchmark:
\begin{itemize}[leftmargin=*,itemsep=1pt,topsep=2pt]
\item the forest's normal approximation and TMLE, because the required convergence-rate condition was not
established for this learner (its variance comparison, which needs only $L^2$ consistency, is made);
\item the product remainder $T_2$ at $\alpha=0.25$, because this learner's errors were designed to decay at exactly
$m^{-1/4}$, the boundary of the product-rate condition of Lemma~\ref{S:lem:prodrate} (its $T_1$, which needs only
consistency, is compared);
\item summaries for which no Monte Carlo interval was formed: the correctly specified learner's $T_2$, and the
bridge to $\Psi$ at fractions above 0.0025.
\end{itemize}
The estimates are reported in all these cases.

Three control conditions check that the diagnostics respond to departures they are meant to detect. The
point-limit control shifts the centre by $0.30\,\swor$; a variance control computes the standard error treating
units as clusters; and the remainder control perturbs the outcome and propensity functions by one aligned logit
shift, calibrated once on the census so that the 0.9 quantile of $|T_2|/\swor$ equals 0.30. Two further variance
controls, expected to stay inside the variance benchmark, are reported descriptively (Table~\ref{S:tab:S-estimates}).
The bridge to $\Psi$ is compared with its benchmark only at the smallest fraction of the vanishing-fraction
trajectory, $f=0.0025$, the smallest fraction considered (Section~\ref{sec:sim-results}).

\begin{figure}[p]
\footnotesize
\centering
\includegraphics[width=0.87\linewidth]{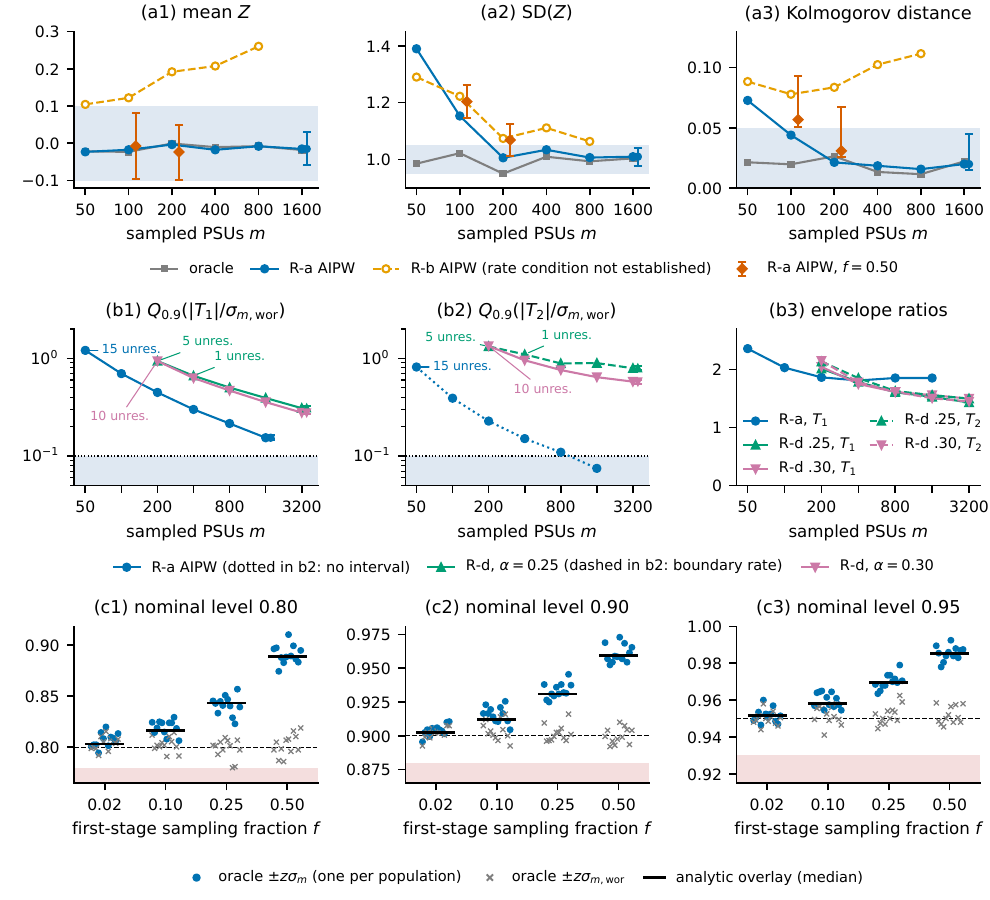}
\caption{Simulation results (estimates in Table~\ref{tab:sim}). (a) Mean, standard deviation and Kolmogorov distance from $\mathcal N(0,1)$ of the standardized error $Z=(\hat\psi-\Psi_N)/\swor$ of cross-fitted AIPW in the many-strata regime ($m_h=2$, $f=0.10$; 400 populations $\times$ 5 samples per point); diamonds show the two configurations with $f=0.50$ (160 $\times$ 5). (b) 0.9 quantiles of $|T_1|/\swor$ and $|T_2|/\swor$ in the fixed-strata regime ($H=5$, $f=0.10$; log scale), and in (b3) of the envelope ratios $\sqrt m|T_1|/\hat r_m$ and $|T_2|/(\hat r_Q\hat r_g)$, with $\hat r$ the measured out-of-fold $L^2$ errors; labels give the number of the 2,000 replicates with an unresolved remainder, which the quantiles exclude. (c) Coverage of the oracle intervals $\hat\psi_{\mathrm{oracle}}\pm z\sigma_m$ (circles) and $\hat\psi_{\mathrm{oracle}}\pm z\swor$ (crosses) across first-stage fractions ($H=50$, $m=100$; 12 populations $\times$ 2,000 samples), with the analytic prediction $2\Phi(z\sigma_m/\swor)-1$ (bars: medians); dashed lines mark the nominal levels. Blue shading marks benchmark regions, and for the Kolmogorov distance the descriptive reference region below 0.05; red, coverage more than 0.02 below nominal. Error bars are 95\% Monte Carlo intervals (for the Kolmogorov distance, resampling summaries only). Hollow symbols mark the forest, for which the required convergence-rate condition was not established; its values are descriptive.}
\label{fig:sim}
\par\smallskip
{\scriptsize\raggedright\noindent\textbf{Alt text:} Nine panels. Top row: the correctly specified estimator approaches the
standard-normal reference as sample size grows, while the forest's standardized mean error (hollow symbols) grows;
with 100 sampled clusters and half of each stratum selected, the fitted estimator's spread exceeds the leading-term
design standard deviation. Middle row: remainder quantiles decline overall but stay above 0.10 at the largest sizes,
except the correctly specified learner's second-order term; envelope ratios level off. Bottom row: oracle coverage
of the 12 populations rises above nominal with the sampling fraction under the ultimate-cluster variance and stays
near nominal under the without-replacement variance.\par}
\end{figure}

\subsection{Results}\label{sec:sim-results}

Table~\ref{tab:sim} reports the principal estimates beside their benchmarks, and Figure~\ref{fig:sim} the
trajectories behind them; Section~\ref{S:supp:simres} gives further results. All \SimCells\ cells completed, and no
learner fit failed. One numerical implementation check failed in a separate code-check cell that contributes no
reported result (Section~\ref{S:supp:experiments}). \SimUnresolvedOther\ AIPW arm-by-replicate remainder values, all
at $m\le400$, could not be resolved numerically. They are recorded as missing and omitted from the remainder
summaries, and the corresponding point and variance estimates are kept. A further \SimUnresolvedTMLE\ unresolved values
concern R-a TMLE, whose remainder calculations were not validated and are not reported at any sample size
(Table~\ref{S:tab:S-unresolved}). None of these \SimUnresolvedTotal\ entries occurs in a configuration whose results are
compared with a benchmark.

\paragraph{Normal approximation.} In the many-strata regime the standardized error of the cross-fitted AIPW
estimator with the correctly specified learner (R-a) met both normal-approximation benchmarks at $m=1600$: mean $Z$
was \SimCoMsMeanZ\ (Monte Carlo interval \SimCoMsMeanZInt) and SD$(Z)-1$ was \SimCoMsSD\ \SimCoMsSDInt. Its
Kolmogorov distance was \SimCoMsKS\ \SimCoMsKSInt, below the descriptive reference 0.05 (Figure~\ref{fig:sim}a). At smaller $m$ the fitted estimator was more
dispersed than its norming (Table~\ref{S:tab:S-ms}).

With half of each stratum's PSUs sampled ($f=0.50$, $m=100$), the same estimator was centred, with mean $Z$
\SimCoFpcMeanZ\ \SimCoFpcMeanZInt, but its standard deviation was \SimCoFpcSDZ\ times $\swor$ (Monte Carlo interval
\SimCoFpcSDZInt), wholly outside its benchmark, and its Kolmogorov distance was \SimCoFpcKS\ \SimCoFpcKSInt, above
the descriptive reference 0.05. At $f=0.50$ and $m=200$ the excess was smaller: SD$(Z)-1$ was \SimCoStrSD\
\SimCoStrSDInt, an interval that crosses its benchmark, so the comparison is unresolved at the simulated precision;
the Kolmogorov distance was \SimCoStrKS\ \SimCoStrKSInt.

For the random forest (R-b), the required convergence-rate condition was not established: the one-sided 95\% upper
bound on the slope of its log error against $\log m$ was \SimRbUpper, not below $-1/4$. Its normal-approximation
results are therefore reported descriptively. Its standardized mean error increased with $m$, from \SimRbMeanZFifty\
at $m=50$ to \SimRbMeanZEight\ at $m=800$, outside the $\pm0.10$ benchmark, while its nominal 95\% interval covered in
\SimRbCovEightPct\% of replicates at $m=800$ (Table~\ref{tab:sim}, panel B). This trend
concerns the error relative to $\swor$, the leading term's design standard deviation, which shrinks as $m$ grows; by
itself it does not show whether the absolute bias increased or decreased.

\paragraph{Coverage and standard-error accuracy.} At $f=0.50$ and $m=100$ the nominal 95\% confidence interval of
the R-a AIPW estimator covered $\Psi_N$ in \SimEtCovNinetyFivePct\% of replicates (Table~\ref{tab:sim}, panel B),
above the nominal level, although the estimator's standard deviation was \SimCoFpcSDZ\ times $\swor$, a variance
excess of about \SimCoFpcVarExcessPct\%. Two effects widen the interval. First, the squared standard error estimates the
ultimate-cluster variance $\sigma_m^2$, which exceeds $\swor^2$ at this fraction: in the 12 oracle populations
evaluated under the same design, $\sigma_m/\swor$ ranged from \SimEtSigRatioMin\ to \SimEtSigRatioMax. Second, the
fitted variance estimate exceeded its own target: the mean of $\widehat{\mathrm{se}}^2/\sigma_m^2$ was \SimEtRvar, a
\SimEtRvarPct\% excess on the variance scale (at most about \SimEtRvarSqrtPct\% on the standard-error scale;
Section~\ref{S:supp:simres}). Coverage above nominal therefore coexisted with excess dispersion and an overestimated
variance; it did not show that the standard error was accurate. The oracle estimator in the same cell had SD$(Z)$
\SimEtOracleSDZ\ and coverage \SimEtOracleCovNinetyFivePct\%, which points to the fitted-nuisance remainders as one
source of the extra spread; the simulation did not isolate that mechanism.

\paragraph{Variance estimation.} The ultimate-cluster standard error met all three variance benchmarks for both
learners at their largest sample sizes. For R-a at $m=1600$ the mean of $R_{\mathrm{var}}$ exceeded 1 by
\SimCtRaMean\ \SimCtRaMeanInt, its standard deviation was \SimCtRaSD, and the mean transfer was \SimCtRaTr. For the
random forest at $m=800$, whose slope diagnostic was consistent with the $L^2$ consistency this comparison needs, the
mean excess was \SimCtRbMean\ \SimCtRbMeanInt, inside the $\pm0.15$ benchmark; a \SimCtRbPct\% variance excess
corresponds to about \SimCtRbSEPct\% on the standard-error scale. At smaller $m$ both fitted learners overestimated
$\sigma_m^2$ by more (Table~\ref{S:tab:S-ms}).

\paragraph{Without-replacement gap and conservativeness.} For the oracle estimator, the without-replacement design
variance agreed with the theory in all \SimCthaN\ populations at $f=0.50$: the signed difference between
$\widehat{\mathrm{Var}}(S_{\mathrm{des}})/\sigma_m^2$ and $1-\Delta_m/\sigma_m^2$ ranged from \SimCthaMin\ to
\SimCthaMax, and all Monte Carlo intervals lay inside $\pm0.05$. Along the fixed-fraction trajectory the census
relative gap $\Delta_m/\sigma_m^2$ stayed between \SimStrictMin\ and \SimStrictMax, which is compatible with
normings that stay strictly separated; a finite trajectory cannot establish an asymptotic limit. Across the four fractions,
the 12 populations and the levels 0.80, 0.90 and 0.95, the largest observed coverage shortfall of the oracle interval
$\hat\psi_{\mathrm{oracle}}\pm z\sigma_m$ below its nominal level was \SimCtbWorstShortfall, against the study's
benchmark of 0.02, and coverage rose with the fraction as the analytic prediction
$2\Phi(z_{1-\gamma/2}\sigma_m/\swor)-1$ indicates (Figure~\ref{fig:sim}c). At $f=0.50$ and level 0.95, coverage
ranged from \SimCtbJFourCovDetMin\ to \SimCtbJFourCovDetMax\ across the 12 populations, with median
\SimCtbJFourCovDetMed\ against \SimCtbJFourOverlayMed\ predicted. Each population's coverage rests on 2,000 redraws,
so its binomial standard error is about \SimCovSEBig; Table~\ref{S:tab:S-fpc} reports results for the 12 populations.

\paragraph{Remainders.} Both remainder terms decreased overall with $m$ (Figure~\ref{fig:sim}b; the only rise in a
displayed trajectory is that of $T_2$ at the boundary rate $\alpha=0.25$, from \SimCfDqTwoEight\ at $m=800$ to
\SimCfDqTwoSixteen\ at $m=1600$, before it fell). The remainder diagnostics assessed against the 0.10 benchmark remained above it
at the largest sample size. For the correctly specified learner, the 0.9 quantile of $|T_1|/\swor$ fell from
\SimCfRaFirst\ at $m=50$ to \SimCfRaQ\ at $m=1600$ (Monte Carlo interval \SimCfRaQInt), by a factor of
\SimCfRaDropMin\ to \SimCfRaDropMax\ per doubling of $m$; its second-order term, summarized without a Monte Carlo
interval, was \SimCfRaTtwoTerm\ at $m=1600$, below 0.10. For the throttled learner at $m=3200$ the 0.9 quantile of
$|T_1|/\swor$ was \SimCfDqT\ \SimCfDqTInt\ at $\alpha=0.25$ and \SimCfDtT\ \SimCfDtTInt\ at $\alpha=0.30$, and that
of $|T_2|/\swor$ at $\alpha=0.30$ fell from \SimCfDtTwoFirst\ at $m=200$ to \SimCfDtTwo\ \SimCfDtTwoInt; at the
boundary rate $\alpha=0.25$ it fell from \SimCfDqTwoFirst\ to \SimCfDqTwo. The envelope ratios
$\sqrt m|T_1|/\hat r_m$ and $|T_2|/(\hat r_Q\hat r_g)$, with $\hat r$ the measured out-of-fold errors, are diagnostics
motivated by the proof's bounds, which use the deterministic sequence $r_m$ rather than $\hat r$; they were stable or
decreasing along each sequence (Figure~\ref{fig:sim}b3).

The remainder quantiles declined overall, but at the largest sample sizes those compared with the benchmark
remained above 0.10, even for the correctly specified learner, whose $T_1$ quantile was about \SimCfRaTimes\ times
the benchmark at $m=1600$. This is a finite-sample finding: the theorem's remainder statements are asymptotic, with
no rate attached, so it neither contradicts nor establishes them. Coverage differed across learners. At $m=1600$ the
correctly specified learner's interval covered in \SimFsRaCovPct\% of replicates (SD$(Z)$ \SimFsRaSDZ), against
\SimFsOrCovPct\% for the oracle. At $m=3200$ the throttled learner, whose remainders were larger, covered in
\SimFsDqCovPct\% and \SimFsDtCovPct\% of replicates at $\alpha=0.25$ and $0.30$ (SD$(Z)$ \SimFsDqSDZ\ and
\SimFsDtSDZ), against \SimFsOrCovThreePct\% for the oracle (Table~\ref{tab:sim}, panel B). These comparisons do not
isolate the contribution of an individual remainder.

No comparison is available below the quarter-rate threshold, where it would have asked whether $T_1$ remains
negligible when the product-rate condition fails: the only schedule there, $\alpha=0.20$, was not run because its
learner's error decayed faster than designed (Table~\ref{S:tab:S-rates}).

\paragraph{TMLE.} Cross-fitted TMLE with the correctly specified learner differed little from the matched AIPW
estimator at $m=800$: the 0.9 quantile of their standardized difference was \SimCsPsi\ \SimCsPsiInt, and that of the
difference in their variance estimates, relative to $\sigma_m^2$, was \SimCsSe\ \SimCsSeInt, both far inside the
0.10 benchmark. The forest TMLE is reported descriptively for the same reason (Table~\ref{tab:sim}, panel B).

\paragraph{Finite-population and super-population targets.} The census parameter approached the super-population
target as the first-stage fraction decreased, as Lemma~\ref{S:lem:modterm} describes, but slowly relative to $\swor$:
the 0.9 quantile across populations of $|\Psi_N-\Psi|/\swor$ was \SimBridgeFirst\ at $f=0.20$ and \SimBridgeLastTraj\
at $f=0.025$, so at $f=0.025$ one population in ten had $\Psi_N$ more than a quarter of $\swor$ from $\Psi$. At
$f=0.0025$ it was \SimBridge\ \SimBridgeInt, inside the 0.10 benchmark. The smallest fraction was included after a leading-order calculation predicted a 0.9 quantile
of about 0.07 there (about 0.22 at $f=0.025$). Table~\ref{S:tab:S-vanish} and Figure~\ref{S:fig:bridge} report
the observed quantiles. This comparison illustrates the prediction rather than testing it independently. The trajectory shows how far apart the two targets can be at sampling fractions
that surveys use, which is why the paper's primary target is $\Psi_N$.

\paragraph{Controls.} The three control conditions reproduced the departures they were designed to produce. The shifted centre gave mean $Z$ \SimPcShift\ \SimPcShiftInt; treating units as clusters gave
mean $R_{\mathrm{var}}$ \SimPcUnit\ \SimPcUnitInt, outside $[0.85,1.15]$; and the aligned nuisance error, calibrated
on the census to 0.30, gave a 0.9 quantile of $|T_2|/\swor$ of \SimPcAlign\ \SimPcAlignInt. These show that the
diagnostics respond to these specific departures at these sizes, not to every possible failure.

{\footnotesize\setlength{\tabcolsep}{3pt}
\begin{longtable}{@{}p{0.23\linewidth}lll@{}}
\caption{Principal simulation estimates with 95\% Monte Carlo intervals, beside the study's benchmarks.
Section~\ref{sec:sim-grading} defines $Z$, $R_{\mathrm{var}}$, the transfer, the benchmarks and the Monte Carlo
intervals. \emph{Panel A.} Arms: R-a, correctly specified learner; R-b, random forest; R-d, throttled learner with
designed error rate $m^{-\alpha}$. KS is the Kolmogorov distance of $Z$ from $\mathcal N(0,1)$; its 0.05 is a
descriptive reference, not a benchmark, and its interval summarizes resampling variability only. $Q_{0.9}|\cdot|$ is the
0.9 quantile of the absolute quantity standardized by $\swor$ (by $\sigma_m^2$ for the difference in variance
estimates), with $\Delta\hat\psi=\hat\psi_{\mathrm{TMLE}}-\hat\psi_{\mathrm{AIPW}}$. The gap error is
$\widehat{\mathrm{Var}}(S_{\mathrm{des}})/\sigma_m^2-(1-\Delta_m/\sigma_m^2)$; its row gives the range over
populations, and the brackets enclose the extreme population-specific Monte Carlo limits. ``Lowest
coverage$-$level'' is the smallest difference, over the 12 populations and the levels 0.80, 0.90 and 0.95, between the
coverage of $\hat\psi_{\mathrm{oracle}}\pm z\sigma_m$ and its nominal level; no interval is given for it. Values
marked ``no interval'' were summarized without a Monte Carlo interval; ``boundary rate'' marks $T_2$ at
$\alpha=0.25$, where the learner's designed error rate is on the boundary of the product-rate condition. Both are
descriptive. The fraction $f=0.0025$ was chosen after a prediction (Section~\ref{sec:sim-results}). Not shown: the forest's normal-approximation and
TMLE values, which are descriptive because its convergence-rate condition was not established (they are in panel
B), and the throttled rates $\alpha=0.20$ and $0.35$, which were not run. Panel B follows.}\label{tab:sim}\\
\toprule
setting & measure & estimate [95\% MC interval] & benchmark\\
\midrule
\endfirsthead
\toprule
setting & measure & estimate [95\% MC interval] & benchmark\\
\midrule
\endhead
\bottomrule
\endlastfoot
\multicolumn{4}{@{}l}{\textit{Normal approximation of $Z$; correctly specified learner (R-a), AIPW}}\\
many strata, $m=1600$ & mean $Z$ & \ensuremath{-0.015} [\ensuremath{-0.059}, \ensuremath{0.029}] & $|\cdot|\le 0.10$\\
 & SD$(Z)-1$ & \ensuremath{0.010} [\ensuremath{-0.022}, \ensuremath{0.042}] & $|\cdot|\le 0.05$\\
 & KS & \ensuremath{0.020} [\ensuremath{0.015}, \ensuremath{0.045}] & $\le 0.05$ (reference)\\
$f=0.50$, $m=100$ & mean $Z$ & \ensuremath{-0.007} [\ensuremath{-0.097}, \ensuremath{0.081}] & $|\cdot|\le 0.10$\\
 & SD$(Z)-1$ & \ensuremath{0.203} [\ensuremath{0.146}, \ensuremath{0.263}] & $|\cdot|\le 0.05$\\
 & KS & \ensuremath{0.057} [\ensuremath{0.051}, \ensuremath{0.093}] & $\le 0.05$ (reference)\\
$f=0.50$, $m=200$ & mean $Z$ & \ensuremath{-0.023} [\ensuremath{-0.099}, \ensuremath{0.048}] & $|\cdot|\le 0.10$\\
 & SD$(Z)-1$ & \ensuremath{0.069} [\ensuremath{0.013}, \ensuremath{0.126}] & $|\cdot|\le 0.05$\\
 & KS & \ensuremath{0.031} [\ensuremath{0.026}, \ensuremath{0.067}] & $\le 0.05$ (reference)\\
\midrule
\multicolumn{4}{@{}l}{\textit{Variance estimation, $R_{\mathrm{var}}=\widehat{\mathrm{se}}^2/\sigma_m^2$; AIPW, many strata}}\\
R-a, $m=1600$ & mean $R_{\mathrm{var}}-1$ & \ensuremath{0.008} [\ensuremath{0.006}, \ensuremath{0.011}] & $|\cdot|\le 0.15$\\
 & SD$(R_{\mathrm{var}})$ & \ensuremath{0.055} [\ensuremath{0.053}, \ensuremath{0.057}] & $\le 0.15$\\
 & transfer & \ensuremath{0.016} [\ensuremath{0.016}, \ensuremath{0.017}] & $\le 0.25$\\
R-b (forest), $m=800$ & mean $R_{\mathrm{var}}-1$ & \ensuremath{0.130} [\ensuremath{0.126}, \ensuremath{0.134}] & $|\cdot|\le 0.15$\\
 & SD$(R_{\mathrm{var}})$ & \ensuremath{0.093} [\ensuremath{0.090}, \ensuremath{0.096}] & $\le 0.15$\\
 & transfer & \ensuremath{0.131} [\ensuremath{0.128}, \ensuremath{0.133}] & $\le 0.25$\\
\midrule
\multicolumn{4}{@{}l}{\textit{Oracle estimator, 12 fixed populations ($H=50$, $m=100$)}}\\
$f=0.50$ & gap error, range & \ensuremath{-0.011} to \ensuremath{0.014} [\ensuremath{-0.028}, \ensuremath{0.033}] & $|\cdot|\le 0.05$\\
$f=0.02$ & lowest coverage$-$level & \ensuremath{-0.006} & $\ge -0.02$\\
$f=0.10$ & lowest coverage$-$level & \ensuremath{0.004} & $\ge -0.02$\\
$f=0.25$ & lowest coverage$-$level & \ensuremath{0.014} & $\ge -0.02$\\
$f=0.50$ & lowest coverage$-$level & \ensuremath{0.028} & $\ge -0.02$\\
\midrule
\multicolumn{4}{@{}l}{\textit{Remainders at the largest sample size, fixed strata ($H=5$); $Q_{0.9}|T_j|$ standardized by $\sigma_{m,\mathrm{wor}}$}}\\
R-a, $m=1600$ & $Q_{0.9}|T_1|$ & \ensuremath{0.154} [\ensuremath{0.148}, \ensuremath{0.163}] & $\le 0.10$\\
 & $Q_{0.9}|T_2|$ & \ensuremath{0.075} (no interval) & $\le 0.10$\\
R-d, $\alpha=0.25$, $m=3200$ & $Q_{0.9}|T_1|$ & \ensuremath{0.309} [\ensuremath{0.297}, \ensuremath{0.323}] & $\le 0.10$\\
 & $Q_{0.9}|T_2|$ & \ensuremath{0.793} [\ensuremath{0.756}, \ensuremath{0.843}] & $\le 0.10$; boundary rate\\
R-d, $\alpha=0.30$, $m=3200$ & $Q_{0.9}|T_1|$ & \ensuremath{0.279} [\ensuremath{0.268}, \ensuremath{0.289}] & $\le 0.10$\\
 & $Q_{0.9}|T_2|$ & \ensuremath{0.578} [\ensuremath{0.550}, \ensuremath{0.605}] & $\le 0.10$\\
\midrule
\multicolumn{4}{@{}l}{\textit{TMLE versus AIPW on the same samples; R-a, many strata, $m=800$}}\\
R-a & $Q_{0.9}|\Delta\hat\psi|$ & \ensuremath{0.0066} [\ensuremath{0.0063}, \ensuremath{0.0072}] & $\le 0.10$\\
 & $Q_{0.9}|\Delta\widehat{\mathrm{se}}^2|$ & \ensuremath{0.0013} [\ensuremath{0.0012}, \ensuremath{0.0014}] & $\le 0.10$\\
\midrule
\multicolumn{4}{@{}l}{\textit{Finite-population versus super-population target (census, $m_h=2$)}}\\*
$f=0.20$ & $Q_{0.9}|\Psi_N-\Psi|$ & \ensuremath{0.704} (no interval) & $\le 0.10$\\*
$f=0.10$ & $Q_{0.9}|\Psi_N-\Psi|$ & \ensuremath{0.440} (no interval) & $\le 0.10$\\*
$f=0.05$ & $Q_{0.9}|\Psi_N-\Psi|$ & \ensuremath{0.326} (no interval) & $\le 0.10$\\*
$f=0.025$ & $Q_{0.9}|\Psi_N-\Psi|$ & \ensuremath{0.261} (no interval) & $\le 0.10$\\*
$f=0.0025$ & $Q_{0.9}|\Psi_N-\Psi|$ & \ensuremath{0.063} [\ensuremath{0.053}, \ensuremath{0.085}] & $\le 0.10$\\
 
\end{longtable}}

{\scriptsize\setlength{\tabcolsep}{4pt}\renewcommand{\theHtable}{panelB.\arabic{table}}
\begin{longtable}{@{}llrrrrr@{}}
\caption*{\textbf{Table~\ref{tab:sim}, panel B.} Coverage of the estimator's nominal 95\% confidence interval
$\hat\psi\pm1.96\,\widehat{\mathrm{se}}$ about $\Psi_N$, reported with its components: mean $Z$ (the standardized
bias), SD$(Z)$ (the spread relative to $\swor$) and mean $R_{\mathrm{var}}$ (the variance estimate relative to its
target $\sigma_m^2$). These are descriptive point estimates over the replicates shown, without Monte Carlo
intervals. For orientation, the binomial standard error of a coverage near 0.95 would be about \SimCovSEBig\ for
2,000 and \SimCovSESmall\ for 800 independent replicates; the replicates here are clustered, five per population, so
these values are approximations that ignore that clustering.}\\
\toprule
setting & arm & mean $Z$ & SD$(Z)$ & mean $R_{\mathrm{var}}$ & coverage & replicates\\
\midrule
\endfirsthead
\toprule
setting & arm & mean $Z$ & SD$(Z)$ & mean $R_{\mathrm{var}}$ & coverage & replicates\\
\midrule
\endhead
\bottomrule
\endlastfoot
many strata, $m=1600$ & oracle & \ensuremath{-0.018} & \ensuremath{1.004} & \ensuremath{1.000} & \ensuremath{0.956} & 2000\\
 & R-a AIPW & \ensuremath{-0.015} & \ensuremath{1.010} & \ensuremath{1.008} & \ensuremath{0.955} & 2000\\
many strata, $m=800$ & oracle & \ensuremath{-0.008} & \ensuremath{0.994} & \ensuremath{0.999} & \ensuremath{0.958} & 2000\\
 & R-a AIPW & \ensuremath{-0.008} & \ensuremath{1.007} & \ensuremath{1.015} & \ensuremath{0.958} & 2000\\
 & R-a TMLE & \ensuremath{-0.008} & \ensuremath{1.006} & \ensuremath{1.015} & \ensuremath{0.958} & 2000\\
 & R-b AIPW & \ensuremath{0.260} & \ensuremath{1.064} & \ensuremath{1.130} & \ensuremath{0.957} & 2000\\
 & R-b TMLE & \ensuremath{0.247} & \ensuremath{1.047} & \ensuremath{1.130} & \ensuremath{0.960} & 2000\\
$f=0.50$, $m=100$ & oracle & \ensuremath{0.012} & \ensuremath{1.045} & \ensuremath{0.996} & \ensuremath{0.983} & 800\\
 & R-a AIPW & \ensuremath{-0.007} & \ensuremath{1.203} & \ensuremath{1.146} & \ensuremath{0.969} & 800\\
 & R-b AIPW & \ensuremath{0.218} & \ensuremath{1.300} & \ensuremath{1.302} & \ensuremath{0.959} & 800\\
$f=0.50$, $m=200$ & oracle & \ensuremath{-0.039} & \ensuremath{0.993} & \ensuremath{1.004} & \ensuremath{0.983} & 800\\
 & R-a AIPW & \ensuremath{-0.023} & \ensuremath{1.069} & \ensuremath{1.084} & \ensuremath{0.975} & 800\\
 & R-b AIPW & \ensuremath{0.171} & \ensuremath{1.206} & \ensuremath{1.239} & \ensuremath{0.965} & 800\\
fixed strata, $m=1600$ & oracle & \ensuremath{-0.001} & \ensuremath{0.976} & \ensuremath{0.998} & \ensuremath{0.962} & 2000\\
 & R-a AIPW & \ensuremath{0.003} & \ensuremath{0.982} & \ensuremath{1.006} & \ensuremath{0.965} & 2000\\
fixed strata, $m=3200$ & oracle & \ensuremath{0.015} & \ensuremath{0.994} & \ensuremath{0.999} & \ensuremath{0.961} & 2000\\
 & R-d AIPW, $\alpha=0.25$ & \ensuremath{0.017} & \ensuremath{1.141} & \ensuremath{1.059} & \ensuremath{0.937} & 2000\\
 & R-d AIPW, $\alpha=0.30$ & \ensuremath{0.001} & \ensuremath{1.073} & \ensuremath{1.043} & \ensuremath{0.945} & 2000\\
 
\end{longtable}}
\addtocounter{table}{-1}

\section{Real-data illustration: uncertainty under an imposed facility-sampling design}\label{sec:app}

For the urban--rural satisfaction comparison introduced in Section~\ref{sec:intro}, we use the SPA surveys of
Malawi 2013--14 and of Haiti 2017--18 and 2013 \citep{MalawiSPA2014,HaitiSPA2019,HaitiSPA2014}, which interviewed clients as they left
antenatal-care, family-planning and sick-child consultations. The substantive question is how the proportion of
clients who report being ``very satisfied'' differs between urban and rural facilities after adjustment for the
client's age group, education and the type of service. Urban location is not a randomized treatment, so this
contrast is an adjusted association, estimated with the paper's influence-function machinery; it is not presented
as a causal effect. The methodological question, and the reason for including the study, is whether the
cross-fitted standard error reports the estimator's actual sampling variability under a known two-stage design with
an appreciable first-stage fraction.

\paragraph{Frames and imposed design.} In each survey we treat the eligible records as a fixed finite population:
completed facility interviews with a completed observation and exit interview for one of the three services, in
facilities with at least one such record, excluding records with a missing exposure or outcome. The frames hold
\AppMwFacilities\ facilities and \AppMwRecords\ records in Malawi (\AppMwExcluded\ records excluded),
\AppHtsFacilities\ facilities and \AppHtsRecords\ records in Haiti 2017--18, and \AppHttFacilities\ facilities and
\AppHttRecords\ records in Haiti 2013. They are analytical populations defined by the surveys' own fieldwork and
eligibility rules, not the countries' full facility lists or client populations, and the original selection of
clients is not reconstructed. On each frame we impose stratified two-stage sampling with simple random sampling
without replacement at both stages, the sampling stages that Theorem~\ref{thm:clt} assumes, and repeat it 500 times.
Facilities are grouped into three facility-type strata; about a quarter of the facilities in each stratum are drawn by
SRS-WOR (\AppMwPsu, \AppHtsPsu\ and \AppHttPsu\ facilities per sample), and then $\min(3,n_k)$ records are drawn by
SRS-WOR within each sampled facility $k$, where $n_k$ is its number of eligible records. The design weight is the
reciprocal of the product of the two inclusion probabilities, so the strata, the inclusion probabilities and the
first-stage fraction are known exactly.

\paragraph{Estimators.} The two primary methods are cross-fitted TMLE and cross-fitted AIPW as implemented in the
\textsf{svytmle} R package, version 0.1.2 \citep{svytmle2026} (see Section~\ref{S:supp:app}).
Both are fitted afresh on each sample. The nuisance models are Super Learner 
ensembles \citep{vanderLaanPolleyHubbard2007,SuperLearner2026} of logistic regression, weighted multivariate adaptive
regression splines, penalized regression and random forests, with fitted propensities truncated to $[0.05,0.95]$.
They are fitted with the survey weights, rescaled within each training fold. This differs deliberately from the
unweighted learners of the simulation: here the second-stage inclusion probability depends on facility caseload,
which plausibly relates to both the exposure and the outcome, so sampling is informative given the adjustment
covariates. Whole facilities were assigned to folds within strata. Within each training fold, Super Learner chooses
its ensemble weights by its default record-level cross-validation, so records from one training facility can fall
on both sides of an inner split; records from the held-out facilities do not enter that fit. Five folds were requested, and the frame counts
imply that five were used throughout (at least \AppMinSampledAll\ facilities were sampled in every stratum), although
the effective number was not logged. The package balances fold sizes only approximately, so the folds do not satisfy
the theorem's exact balance condition $V\mid m_h$; Section~\ref{S:supp:app} gives the fold rule.
The TMLE applies one targeting fluctuation pooled over folds and weighted by the design weights. The standard error
is the stratified ultimate-cluster estimator on the fitted influence function, without a finite-population
correction. To evaluate the intervals as the package reports them, we retain its Student-$t$ reference with design degrees
of freedom equal to the number of sampled facilities minus the number of strata (180 to 199 here), the
conventional choice when the variance is estimated from between-facility variation within strata
\citep[Section~5.2]{KornGraubard1999}; the simulation uses the standard normal reference of the theorem's limit. 
The $t$ critical values, about 1.97
rather than 1.96, widen the intervals by about 0.6\%--0.7\%. The reference affects interval width and coverage,
but not the variance ratios or the standard-error ratios.

\paragraph{Evaluation.} The primary measure is the variance ratio: the mean of the 500 estimated variances divided by
the empirical variance of the 500 estimates, with a replicate-bootstrap Monte Carlo standard error. Because the
resamples within each survey are drawn from the same frame, the empirical variance estimates the implemented estimator's sampling variance
under the imposed design, so the ratio needs no knowledge of a true effect. A ratio above one means that the reported
uncertainty is too large on average; it does not identify the source of the excess. We also report the mean standard
error divided by the empirical standard deviation, which is not the square root of the variance ratio. Coverage needs
a reference value, which a frame does not supply. Each estimator family was therefore also run on the whole frame
with unit weights over 20 fold splits, and the median serves as that family's whole-frame reference value
(Table~\ref{S:tab:S-app-ref}). Coverage of this algorithmic reference is a secondary, descriptive measure; it is not
coverage of a known causal effect or of a known value of $\Psi_N$. Malawi is the primary survey and the two Haiti
surveys are replications. Because the ultimate-cluster standard error is expected to be conservative at a
first-stage fraction of 0.25, we expected both primary methods to have a variance ratio of at least one minus two
Monte Carlo standard errors.

\paragraph{Findings.} Table~\ref{tab:app} reports all three surveys and both primary estimators;
all 500 samples succeeded for each method. The adjusted urban-minus-rural differences in the
proportion very satisfied, using the cross-fitted TMLE whole-frame reference values, were
\AppMwTmleRef\ percentage points in Malawi, \AppHtsTmleRef\ in Haiti 2017--18 and
\AppHttTmleRef\ in Haiti 2013. All six primary variance ratios exceeded one, from
\AppHttAipwVr\ to \AppMwTmleVr\ (Figure~\ref{fig:app}b), so both primary methods met the stated expectation in all
three surveys. That expectation is a lower bound and does not require the estimated and
empirical variances to agree. For AIPW in Haiti 2013 the ratio, \AppHttAipwVr\ (Monte Carlo
standard error \AppHttAipwVrMcse), is within two Monte Carlo standard errors of one, so
conservativeness is not resolved there. The largest excess was for TMLE in Malawi,
\AppMwTmleVr\ (Monte Carlo standard error \AppMwTmleVrMcse). Mean reported standard errors
were \AppSePctMin\% to \AppSePctMax\% larger than the empirical standard deviations of the
estimates. Figure~\ref{fig:app}a illustrates the largest of these discrepancies: the Malawi TMLE estimates had
empirical standard deviation \AppMwTmleSdPp\ percentage points, against a mean reported standard error of
\AppMwTmleMeanSePp. Coverage of the whole-frame reference values ranged from \AppCovMin\% to
\AppCovMax\%, with Monte Carlo standard errors of \AppCovMcseMin\ to \AppCovMcseMax\
percentage points.

How much of the excess could the omitted finite-population correction account for? With about a quarter of the
facilities sampled in every stratum, the ultimate-cluster variance target can exceed the without-replacement variance
of the leading term by a factor of at most about $1/(1-0.25)\approx$ \AppFpcBound\ (dashed line in
Figure~\ref{fig:app}b; about \AppFpcBoundSe\ on the standard-error scale), attained only if all variation were between facilities (Theorem~\ref{thm:clt}). Five of the
six primary variance ratios were below this value. The Malawi TMLE ratio exceeded it by \AppMwTmleBoundZ\ Monte Carlo
standard errors, so the correction alone cannot explain that excess. In Malawi the two primary methods reported
nearly the same mean interval widths. The ratios of mean standard error to empirical standard deviation were
\AppMwTmleSeSd\ for TMLE and \AppMwAipwSeSd\ for AIPW (Table~\ref{S:tab:S-app-arms}).
With nearly equal reported standard errors, the larger TMLE ratio reflects less variation in its point estimates.
The comparison is descriptive: the frames do not reveal the between-facility share of variance, and the empirical
variance also contains the fitted-nuisance remainders, whose contribution was not isolated.

\begin{figure}[tbp]
\footnotesize
\centering
\includegraphics{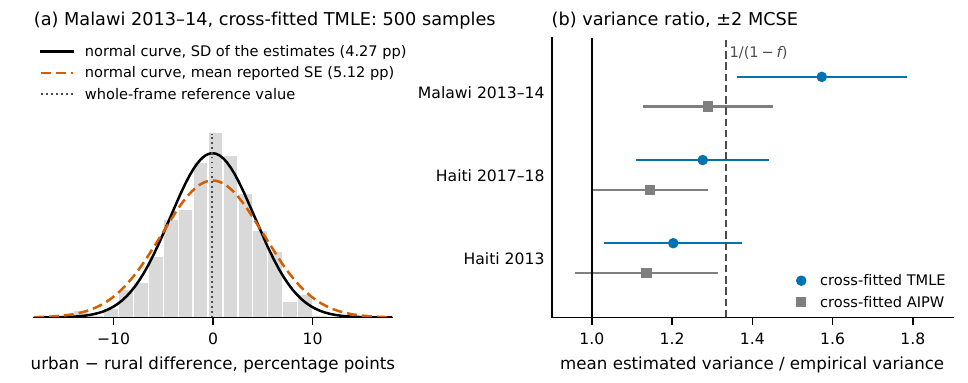}
\caption{Reported and empirical uncertainty in the SPA census-resampling study: 500 samples per survey, each drawing
approximately one quarter of the facilities in each stratum without replacement, then up to three eligible records
per sampled facility. (a) Malawi cross-fitted TMLE estimates of the adjusted urban-minus-rural difference in the
proportion ``very satisfied'', in percentage points (pp). The two normal curves share the mean of the estimates and
use the empirical standard deviation (\AppMwTmleSdPp~pp) and the mean reported standard error
(\AppMwTmleMeanSePp~pp); they illustrate the two spread measures and are not normality diagnostics or confidence
intervals. The dotted line is the algorithmic whole-frame reference value (median of 20 whole-frame fits). (b) Mean
estimated variance divided by empirical variance, for both primary estimators in the three surveys; bars are $\pm2$
replicate-bootstrap Monte Carlo standard errors (MCSE). The solid line marks agreement at 1. The dashed line,
$1/(1-0.25)\approx\AppFpcBound$, is an approximate benchmark at the nominal fraction $f=0.25$ for the largest ratio
$\sigma_m^2/\swor^2$ of the two leading-term variance scales that omitting the first-stage finite-population
correction can produce; it is neither an exact bound from the realized stratum fractions nor a bound on the fitted
estimator's empirical ratio. The SE/SD ratio of Table~\ref{tab:app} differs from the variance ratio in scale and
because the mean squared standard error is not the square of the mean standard error.}
\label{fig:app}
\par\smallskip
{\scriptsize\raggedright\noindent\textbf{Alt text:} (a) Histogram of the 500 Malawi cross-fitted TMLE estimates,
centred near zero, with two normal curves: a narrower solid one (standard deviation of the estimates) and a wider
dashed one (mean reported standard error). (b) Six variance ratios with error bars, for two estimators in three
surveys, all above 1; only the Malawi TMLE ratio lies above the dashed line at 1.33.\par}
\end{figure}

\paragraph{Sensitivity and limitations.} The sensitivity versions use unweighted nuisances and add
facility type and log caseload to the covariates, so they target a differently adjusted contrast, and differences
from the primary versions cannot be attributed to the nuisance weighting alone. Their calibration was not uniform
across frames (Table~\ref{S:tab:S-app-arms}). In Malawi the sensitivity TMLE underestimated the variance
(ratio \AppMwSensTmleVr) and covered its whole-frame reference value in \AppMwSensTmleCov\% of samples; \AppMwOverlapSens\% of the
estimated census propensities for that covariate set fell outside $[0.05,0.95]$, above the 10\% share at which a
covariate set is flagged as overlap-strained (Haiti \AppHtsOverlapSens\% and \AppHttOverlapSens\%; primary covariates \AppMwOverlapPrim\% in
all three frames). A separate Malawi overlap diagnostic logged a learner error on rare education levels and may have
used a reduced ensemble; it feeds into no estimate. The two Haiti surveys are separate frames from one country, not
independent replications across countries. More broadly, the study evaluates the implemented algorithms under one
imposed design (a quarter of facilities, at most three records each) on observed frames. It does not establish the
theorem's model-layer or nuisance-rate conditions, calibration under the surveys' original designs, or behaviour at
other fractions; Section~\ref{S:supp:app} describes the methods and gives the design diagnostics and the whole-frame reference values.

For an analyst, coverage of the whole-frame reference values near or above 95\% did not reveal how much the reported uncertainty
exceeded the estimates' actual spread; the variance ratio did.

{\small\setlength{\tabcolsep}{2.5pt}
\begin{longtable}{@{}llrrrr@{}}
\caption{Uncertainty of the two primary cross-fitted estimators in the SPA census-resampling study (500 samples
per survey under the imposed design: a quarter of facilities per facility-type stratum by SRS-WOR, then up to three
records per facility by SRS-WOR). The variance ratio is the mean estimated variance divided by the empirical variance
of the estimates across samples, with its Monte Carlo standard error (MCSE); SE/SD is the mean estimated standard
error divided by the empirical standard deviation. Coverage is that of the package's nominal 95\% confidence
interval, $\hat\psi\pm t_{0.975,d}\,\widehat{\mathrm{se}}$ with $d$ the design degrees of freedom (sampled facilities
minus strata), for each method family's whole-frame reference value (the median of 20 whole-frame fits), not for a known
effect; its MCSE treats the reference value as fixed. The reference column gives that adjusted
urban-minus-rural difference in the proportion ``very satisfied'', in percentage points (pp). Malawi, the primary
survey, comes first; all eight methods are in Table~\ref{S:tab:S-app-arms}.}\label{tab:app}\\
\toprule
survey & method & variance ratio (MCSE) & SE/SD & coverage, \% (MCSE) & reference, pp\\
\midrule
\endfirsthead
\toprule
survey & method & variance ratio (MCSE) & SE/SD & coverage, \% (MCSE) & reference, pp\\
\midrule
\endhead
\bottomrule
\endlastfoot
Malawi 2013--14 & cross-fitted TMLE & \ensuremath{1.573} (\ensuremath{0.106}) & \ensuremath{1.199} & 95.2 (1.0) & \ensuremath{-0.10}\\
 & cross-fitted AIPW & \ensuremath{1.289} (\ensuremath{0.081}) & \ensuremath{1.089} & 94.4 (1.0) & \ensuremath{-0.03}\\
Haiti 2017--18 & cross-fitted TMLE & \ensuremath{1.276} (\ensuremath{0.083}) & \ensuremath{1.101} & 97.4 (0.7) & \ensuremath{-4.31}\\
 & cross-fitted AIPW & \ensuremath{1.145} (\ensuremath{0.072}) & \ensuremath{1.042} & 96.2 (0.9) & \ensuremath{-4.29}\\
Haiti 2013 & cross-fitted TMLE & \ensuremath{1.203} (\ensuremath{0.086}) & \ensuremath{1.077} & 96.6 (0.8) & \ensuremath{-2.96}\\
 & cross-fitted AIPW & \ensuremath{1.136} (\ensuremath{0.089}) & \ensuremath{1.043} & 96.2 (0.9) & \ensuremath{-2.96}\\
 
\end{longtable}}

\section{Discussion}\label{sec:disc}

\paragraph{Main contribution.}
Under the conditions of Theorem~\ref{thm:clt}, survey-weighted cross-fitting with folds formed from whole PSUs within
strata supports inference for a finite-population contrast under first-stage SRS-WOR. The usual stratified
ultimate-cluster variance estimator, computed without a finite-population correction, estimates a with-replacement
variance target. Its relation to the without-replacement variance of the leading term makes the interval
asymptotically conservative for $\Psi_N$ when the first-stage fraction is appreciable, and asymptotically exact when
the relative gap between the two vanishes. Complement stability localizes the control of the empirical-process term
to the single realized fold-specific score, which is why no entropy condition on the learners is needed. The
nuisance-rate and population-model conditions remain essential: the result is not a guarantee for arbitrary
learners or survey designs.

\paragraph{Finite-sample evidence.}
The simulation qualifies how the theorem should be used at finite sample sizes. Each remainder diagnostic assessed against the
study's negligibility benchmark exceeded it at the largest sample size, although the correctly specified learner's
second-order term, summarized without an interval, fell below it. With the correctly specified learner, coverage in the
fixed-strata regime was close to the oracle's; with the slower throttled learners it was \SimFsCovGapMin\ to
\SimFsCovGapMax\ percentage points below the oracle's. For the throttled learner at $\alpha=0.30$ the benchmark was demanding by
design: had the product remainder followed its envelope, its standardized 0.9 quantile would have fallen by only
about a quarter between $m=200$ and 3200; it fell by \SimCfDtTwoDropPct\%, faster than the envelope, yet ended at
\SimCfDtTwoTimes\ times the benchmark. The random forest's standardized mean error increased over the sample sizes
studied while its coverage stayed near nominal, a pattern that coverage alone would not reveal; the required
convergence-rate condition was not established for this learner. An asymptotic approximation can thus be useful in some
configurations without being uniformly accurate across sample sizes and nuisance fits.

\paragraph{Implications for practice.}
Both empirical studies point to the same practical issue at appreciable first-stage fractions. In the simulation
configuration with half of each stratum's PSUs sampled and 100 PSUs, the fitted estimator's spread exceeded $\swor$
by about one fifth, while the ultimate-cluster standard error, built on the larger $\sigma_m$ and overestimating even
that target, made the interval over-cover. In the SPA frames, with a quarter of facilities sampled, the reported
standard errors exceeded the estimates' empirical standard deviations by \AppSePctMin\% to \AppSePctMax\% on average.
These observations concern the configurations studied and are not guarantees. Under the theorem's conditions, the
ultimate-cluster variance targets a quantity at least as large as the leading term's design variance; the
observations add that near-nominal or higher coverage is not evidence that the standard error is accurate. From a single survey the analyst cannot observe this excess. Under SRS-WOR the ultimate-cluster
variance target exceeds the without-replacement variance of the leading term by at most the factor
$1/(1-\max_h f_h)$.

\paragraph{Limitations and future work.}
Several limits apply.

\emph{Design.} The cross-fitted result is established under first-stage SRS-WOR with bounded
cluster sizes. There, the held-out fold, conditioned on the training PSUs, admits the poststratified-complement
design centring. Extending the result to unequal-probability first stages, where that centring can fail, is open,
and the result is not intended to justify applications under such designs. The application's folds are balanced only
approximately, outside the scope proved here (Section~\ref{sec:app}).

\emph{Nuisance and model assumptions.} Complement stability \eqref{eq:cs} is formulated for the realized
fold-specific score and proved from the stratified-mixture model (Condition~\ref{cond:mix}). We do not relate it to
a Donsker condition in either direction or exhibit a non-Donsker class satisfying it. Cross-fitting removes the need for an entropy condition but not for nuisance rates.
Lemma~\ref{S:lem:prodrate} controls $T_2$ through the product of the outcome and propensity errors;
the separate $\od_p(m^{-1/4})$ rates assumed here suffice for that product to be $\od_p(m^{-1/2})$. $\D$ is the efficient influence function of the average treatment
effect only in the nonparametric model for a single unit's data, and no efficiency claim is made under the clustered
sampling model.

\emph{Empirical evidence.} The simulation does not exercise the broader fixed-function designs, the combined regime, or the
qualification that positive between-PSU variance at a fixed fraction need not keep the two normings strictly
separated. The empirical studies compare variants within the AIPW and TMLE family rather than estimators from other
families of survey-adjusted methods. Explanations of the finite-sample discrepancies reported here are left for
future work.

\section*{Supplementary material}
The supplement contains the proofs, the simulation's methods and additional results, and the application's details.

\section*{Data and code availability}
The simulation and SPA analysis code, the locked simulation environment and pinned SPA package
versions, the exact \textsf{svytmle} build, saved summary results and aggregate SPA outputs,
and scripts that regenerate the displays are available in the versioned GitHub release
\texttt{v1.1.1} of the reproduction repository \citep{KarimCode2026}.
The repository README describes the files, the two release archives of saved simulation outputs and the 
reproduction commands. The SPA microdata are available through registration and approval from the Demographic and
Health Surveys (DHS) Program and are not redistributed. The analysis and resampling code is included.

\end{document}